%% file: main.tex
\documentclass[acmsmall,nonacm,screen]{acmart}

\acmJournal{PACMPL}
\acmVolume{6} 
\acmNumber{OOPSLA2}
\acmYear{2022}
\acmSubmissionID{oopslab22main-p505-p}
\acmMonth{10}
\acmPrice{}
\acmDOI{10.1145/3563341}
\startPage{1}

\setcopyright{rightsretained}
\copyrightyear{2022}           

\usepackage{booktabs}   
\usepackage{subcaption} 

\usepackage{tabto}
\usepackage[ruled, linesnumbered]{algorithm2e}
\usepackage[noend]{algpseudocode}
\usepackage{amsmath}
\usepackage{thmtools}
\usepackage{tcolorbox}
\usepackage{arydshln}
\usepackage{xcolor}
\usepackage{listings}
\usepackage{lstautogobble}
\usepackage{siunitx}

\definecolor{keywordsColor}{RGB}{97, 0, 71}
\definecolor{commentsColor}{RGB}{54, 54, 54}

\lstdefinelanguage{OurLanguage}{
    alsoletter={:,=, <, >, &, |},
    keywords={while, end, types, if, else, else:, and, or, not},
    morekeywords={=, <, >, <=, >=, ==, !=, &&, ||},
    basicstyle={\ttfamily\small\normalfont},
    keywordstyle={\color{keywordsColor}\ttfamily\bfseries},
    comment=[l]{\#},
    commentstyle={\color{commentsColor}\ttfamily},
    autogobble=true,
    mathescape=true
}
\lstdefinestyle{program}{basicstyle=\small\ttfamily,keywordstyle=\bfseries}
\usepackage{enumitem}
\usepackage{xargs}
\usepackage{pifont}
\usepackage{nicefrac}
\usepackage{syntax}

\newcommand{\cmark}{\ding{51}}
\newcommand{\xmark}{\ding{55}}

\newcommand{\Mora}{{\tt Mora}}
\newcommand{\Tool}{{\tt Polar}}

\newtheorem{definition}{Definition}
\newtheorem{theorem}{Theorem}
\newtheorem*{theorem*}{Theorem}
\newtheorem{lemma}[theorem]{Lemma}
\newtheorem*{lemma*}{Lemma}

\newtheorem{example}{Example}
\newtheorem*{remark}{Remark}

\newcommand{\R}{{\mathbb{R}}}
\newcommand{\Complex}{{\mathbb{C}}}
\newcommand{\N}{{\mathbb{N}}}
\newcommand{\Z}{{\mathbb{Z}}}
\newcommand{\Normal}{{\text{Normal}}}

\newcommand{\E}{{\mathbb{E}}} 
\renewcommand{\P}{{\mathbb{P}}} 
\newcommand{\Program}{{\mathcal{P}}}

\newcommand{\States}{\text{States}}
\newcommand{\Runs}{\text{Runs}}

\newcommand{\Vars}{\text{Vars}}

\newcommand{\supp}{\text{supp}}

\newcommand{\transmap}{\text{\ }\mapsto\text{\ }}

\begin{document}

\title[This Is the Moment for Probabilistic Loops]{This Is the Moment for Probabilistic Loops}         


\author{Marcel Moosbrugger}
\orcid{0000-0002-2006-3741}             
\affiliation{
  \institution{TU Wien}            
  \city{Vienna}
  \country{Austria}                    
}
\email{marcel.moosbrugger@tuwien.ac.at}          

\author{Miroslav Stankovi\v{c}}
\orcid{0000-0001-5978-7475}             
\affiliation{
  \institution{TU Wien}            
  \city{Vienna}
  \country{Austria}                    
}
\email{miroslav.stankovic@tuwien.ac.at}          

\author{Ezio Bartocci}
\orcid{0000-0002-8004-6601}             
\affiliation{
  \institution{TU Wien}            
  \city{Vienna}
  \country{Austria}                    
}
\email{ezio.bartocci@tuwien.ac.at}          

\author{Laura Kov\'acs}
\orcid{0000-0002-8299-2714}             
\affiliation{
  \institution{TU Wien}            
  \city{Vienna}
  \country{Austria}                    
}
\email{laura.kovacs@tuwien.ac.at}          

\begin{abstract}
We present a novel static analysis technique to derive higher moments for program variables for a large class of probabilistic loops with potentially uncountable state spaces. Our approach is fully automatic, meaning it does not rely on externally provided invariants or templates.
We employ algebraic techniques based on linear recurrences and introduce program transformations to simplify probabilistic programs while preserving their statistical properties. 
We develop power reduction techniques to further simplify the polynomial arithmetic of probabilistic programs and define the theory of moment-computable probabilistic loops for which higher moments can precisely be computed. Our work has applications towards recovering probability distributions of random variables and computing tail probabilities. The empirical evaluation of our results demonstrates the applicability of our work on many challenging examples.
\end{abstract}



\maketitle

\input{01-introduction}

\input{02-preliminaries}

\input{03-model}

\input{04-types}

\input{05-computing-moments}

\input{06-use-cases}

\input{07-experiments}

\input{08-related}

\input{09-conclusion}

\section*{Data Availability Statement}
The tool \Tool{} together with all benchmarks and scripts necessary to reproduce the results reported in this paper are available through an openly accessible artifact~\cite{PolarArtifact}.

\begin{acks}
  This research was supported by the WWTF ICT19-018 grant ProbInG, the ERC Consolidator Grant ARTIST 101002685, the Austrian FWF project W1255-N23, and the SecInt Doctoral College funded by TU Wien.
  We thank the anonymous reviewers for their outstanding and detailed feedback.
\end{acks}

\newpage
\bibliography{references}

\end{document}

%% file: 01-introduction.tex
\section{Introduction}\label{sec:introduction}

Probabilistic programming languages enrich classical imperative or functional languages with native primitives to draw samples from random distributions, such as Bernoulli, Uniform, and Normal distributions. 
The resulting probabilistic programs (PPs)~\cite{Kozen1985,Barthe2020} embed uncertain quantities, represented by random variables, within standard program control flows.
As such, PPs offer a unifying framework to naturally encode probabilistic machine learning models~\cite{Ghahramani2015}, for example Bayesian networks~\cite{Kaminski2016}, into programs. Moreover, PPs enable programmers to handle uncertainty resulting from sensor measurements and environmental perturbations in cyber-physical systems~\cite{Selyunin2015,Chou2020}.
Other notable examples of PPs include the implementation of cryptographic~\cite{Barthe2012} and privacy~\cite{Barthe2012a} protocols, as well as randomized algorithms~\cite{Motwani1995} such as 
Herman's self-stabilization protocol~\cite{Herman1990} for recovering from faults in a process token ring --- see our example in Figure~\ref{fig:herman3}.

\paragraph{Analysis of PPs.}
The random nature of PPs makes their functional analysis very challenging as one needs to reason about probability distributions of random variables instead of computing with single variable values~\cite{Barthe2020}.
A standard approach towards handling probability distributions associated with random variables is to estimate such distributions by sampling PPs using Monte Carlo simulation techniques~\cite{Hastings1970}.
While such approaches work well for statistical model checking~\cite{Younes2006}, they are not suitable for the analysis of PPs with potentially infinite program loops as simulating infinite-state behavior is not always viable.
Moreover, even for PPs with finitely many states, simulation-based analysis is inherently approximative.

With the aim of precisely, and not just approximately, handling random variables, 
probabilistic model checking~\cite{Kwiatkowska2011,Dehnert2017} became a prominent approach in the analysis of PPs with finite state spaces.
For analyzing unbounded PPs, these techniques would however require non-trivial user guidance, in terms of assertion templates and/or invariants.

\emph{In this paper, we address the challenge of precisely analyzing, and even recovering, probability distributions induced by PPs with both countably and uncountably infinite state spaces.} 
We do so by extending both expressivity and automation of the state-of-the-art in PP analysis: We (i) focus on PPs with probabilistic infinite loops (see Figure~\ref{fig:running-example}) and (ii) fully automate the analysis of such loops by computing exact higher-order statistical moments of program variables $x$ parameterized by a loop counter $n$.

Functional representations $f(n)$ for a program variable $x$, with $f(n)$ characterizing the $k$th moment $\E(x^k_n)$ of $x$ at iteration $n$, can be interpreted as a quantitative invariant $\E(x^k_n) - f(n) = 0$, as the equation is true for all loop iterations $n \in \N$.
Inferring quantitative invariants is arguably not novel.
On the contrary, it is one of the most challenging aspects of PP analysis, dating back to the seminal works of~\cite{McIver2005,Katoen2010} introducing the weakest pre-expectations calculus.
Template-based approaches to discover invariants or (super-)martingales emerged~\cite{Barthe2016,Kura2019} by translating the invariant generation problem into a constraint solving one.
The derived quantitative invariants are generally provided in terms of expected values~\cite{Chakarov2014,Katoen2010,McIver2005}. 
Nevertheless, the expected value alone --- also referred to as the \emph{first moment} --- provides only partial information about the underlying probability distribution.
This motivates the critical importance of higher moments for PP analysis~\cite{Kura2019,Bartocci2020a,Wang2021,Stankovic2022}.

\paragraph{Higher Moments for PP Analysis.}
Using concentration-of-measure inequalities~\cite{Boucheron2013}, we can utilize higher moments $\E(X^k)$ to obtain upper and lower bounds on tail probabilities $\P(X > t)$, measuring the probability that a given random variable $X$, corresponding for example to our program variables $x$ from Figure~\ref{fig:running-example}, surpasses some value $t$.
{\it In this paper, we also show that when a program variable $x$ admits only $k < \infty$ many values, we can fully recover its probability mass function as a closed-form expression in the loop counter $n$ using the first $k{-}1$ raw moments} (see Section~\ref{sec:moments}).
Furthermore, raw moments can be used to compute central moments $\E((X - \E(X))^k)$ and thus provide insights on other important characteristics of the distribution such as the \emph{variance}, \emph{skewness} and \emph{kurtosis}~\cite{Durrett2019}.
However, {\it computing exact higher statistical moments} for PPs is computationally expensive~\cite{Kaminski2019}, a challenge which we also {\it address in this paper}, as illustrated in Figures~\ref{fig:running-example}--\ref{fig:herman3} and described next.

\begin{figure}
  \centering
  \begin{minipage}{0.46\textwidth}
    \begin{lstlisting}[basicstyle=\footnotesize]
    toggle,sum,x,y,z = 0,$s_0$,1,1,1
    while $\star$:
      toggle = 1-toggle
      if toggle == 0:
        x = x+1 {1/2} x+2
        y = y+z+x**2 {1/3} y-z-x
        z = z+y {1/4} z-y
      end
      l,g = Laplace(x+y, 1),Normal(0,1)
      if g < 1/2: sum = sum+x end
    end
    \end{lstlisting}
  \end{minipage}
  \begin{minipage}{0.53\textwidth}
    \small
    \begin{tcolorbox}[width=0.9\linewidth,boxrule=1pt,leftrule=2pt,arc=0pt,auto outer arc]  
      $\E(\text{toggle}_n) = \frac{1}{2} - \frac{(-1)^n}{2}$ \\
      \ \\
      $\E(x_n) = \frac{5}{8} + \frac{3n}{4} + \frac{3(-1)^n}{8}$ \\
      \ \\
      $\E(x_n^2) = \frac{15}{32} + \frac{17n}{16} + \frac{9n(-1)^n}{16} + \frac{17(-1)^n}{32} + \frac{9n^2}{16}$ \\
      \ \\
      $\E(l_n) = \frac{-17}{8} - \frac{15n}{4} + \frac{67 \cdot 2^{-n} \cdot 6^{\frac{n}{2}}}{10} + \frac{67 \cdot 2^{-n} 6^{\frac{1 + n}{2}}}{30} -$\\
      $\text{\ \ \ }\frac{37 \cdot 3^{-n} 6^{\frac{n}{2}}}{10} - \frac{37 \cdot 3^{-n} 6^{\frac{1 + n}{2}}}{20} + \frac{67 \cdot 6^{\frac{n}{2}} (-1)^n}{10 \cdot 2^n} -$\\
      $\text{\ \ \ }\frac{15 (-1)^n}{8} - \frac{67 \cdot 6^{\frac{1 + n}{2}} (-1)^n}{30 \cdot 2^n} + \frac{37 \cdot 6^{\frac{1+n}{2}} (-1)^n}{20 \cdot 3^n} -$
      $\text{\ \ \ }\frac{37 \cdot 6^{\frac{n}{2}} (-1)^n}{10 \cdot 3^n}$
    \end{tcolorbox}
  \end{minipage}
  \caption{An example of a multi-path PP loop, with Laplace and Normal distributions parametrized by program variables. Our work fully automates the analysis of such and similar PP loops by computing higher moments. Several moments for program variables in the loop counter $n$ are listed on the right. Each moment was automatically generated.}
  \label{fig:running-example}
\end{figure}

\paragraph{Computing Higher Moments.}
The theory we establish in this paper describes how to compute higher moments of program variables for a large class of probabilistic loops and how to utilize these moments to gain more insights into the analyzed programs. We call this theory the \emph{theory of moment-computable probabilistic loops} (Section~\ref{sec:computing-moments}). 
Our approach is fully automatic, meaning it does \emph{not} rely on externally provided invariants or templates.
Unlike constraint solving over templates~~\cite{Barthe2016,Kura2019}, we employ algebraic techniques based on systems of \emph{linear recurrences} with constant coefficients describing so-called \emph{C-finite sequences}~\cite{Kauers2011}.
Different equivalence preserving program transformations (Section~\ref{sec:model}) and power reduction of finite valued variables (Section~\ref{sec:types}) allow us to simplify PPs and represent their higher moments as linear recurrence systems in the loop counter.
Figure~\ref{fig:running-example} shows a PP with many unique features supported by our work towards PP analysis: it has an uncountable state-space, contains if-statements, symbolic constants, draws from continuous probability distributions with state-dependent parameters, and employs polynomial arithmetic as well as circular variable dependencies. {\it We are not aware of other works automating the reasoning about such and similar probabilistic loops}, in particular for computing precise higher moments of variables. Figure~\ref{fig:running-example} lists some of the variables' moments computed automatically by our work. 
Further, these moments can be used to compute tail probability bounds or central moments, such as the variance, to characterize the distribution of the program variables as the loop progresses.

Thanks to our power reduction techniques (Section~\ref{sec:types}), our approach supports arbitrary polynomial dependencies among finite valued variables.
Moreover, our work can fully recover the value distributions of finite valued program variables, from finitely many higher moments, as illustrated in Section~\ref{sec:moments} for Herman's self-stabilization algorithm from Figure~\ref{fig:herman3}.

\paragraph{Theory and Practice in Computing Higher Moments.} In theory, our approach can compute any higher moment for any variable and PP of our program model, under assumptions stated in Sections~\ref{sec:model} and \ref{sec:computing-moments}. 
We also establish the necessity of these assumtions in Section~\ref{ssec:necessity}. 
In a nutshell, the completeness theorem (Theorem~\ref{thm:moment-computability}) holds for probabilistic loops for which non-finite program variables are not polynomially self-dependent and all branching conditions are over finite valued variables.
We strengthen the theory of~\cite{Bartocci2019} to support if-statements, circular variable dependencies, state-dependent distribution parameters, simultaneous assignments, and multiple assignments, and establish the necessity of our assumptions.
Moreover, unlike~\cite{Wang2021}, our approach does not rely on templates and provides exact closed-form representations of higher moments parameterized by the loop counter.

In practice, our approach is implemented in the \Tool{} tool and compared against exact as well as approximate methods~\cite{PolarArtifact}.
Our experiments (Section~\ref{sec:evaluation}) show that \Tool{} outperforms the state-of-the-art of moment computation for probabilistic loops in terms of supported programs and efficiency.
Furthermore, \Tool{} is able to compute exact higher moments magnitudes faster than sampling can establish reasonable confidence intervals.

\paragraph{Contributions.}
Our main contributions are listed below: 
\begin{itemize}
  \item An automated approach for computing higher moments of program variables for a large class of probabilistic loops with potentially uncountable state spaces (Sections \ref{sec:model}-\ref{sec:computing-moments}).
  \item We develop power reduction techniques to reduce the degrees of finite valued program variables in polynomials (Section~\ref{sec:types}).
  \item We prove completeness of our work for computing higher moments (Section \ref{sec:computing-moments}).
  \item We fully recover the distributions of finite valued program variables and approximate distributions for unbounded/continuous program variables from finitely many moments (Section~\ref{sec:moments}).
  \item We provide an implementation and empirical evaluation of our work, outperforming the state-of-the-art in PP analysis in terms of automation and expressivity (Section~\ref{sec:evaluation}).
\end{itemize}

\begin{figure}
  \centering
  \begin{minipage}{0.5\textwidth}
    \begin{lstlisting}[basicstyle=\footnotesize]
    x1, x2, x3 = 1, 1, 1
    t1, t2, t3 = 1, 1, 1
    p = 1/2; tokens = t1 + t2 + t3
    while $\star$:
      x1o, x2o, x3o = x1, x2, x3
      if x1o == x3o: x1=Bernoulli(p) else: x1=x3o end
      if x2o == x1o: x2=Bernoulli(p) else: x2=x1o end
      if x3o == x2o: x3=Bernoulli(p) else: x3=x2o end
      
      if x1 == x3: t1 = 1 else: t1 = 0
      if x2 == x1: t2 = 1 else: t2 = 0
      if x3 == x2: t3 = 1 else: t3 = 0
      tokens = t1 + t2 + t3
    end
    \end{lstlisting}
  \end{minipage}
  \small
  \begin{minipage}{\linewidth}
    \begin{tcolorbox}[width=\textwidth,boxrule=1pt,leftrule=2pt,arc=0pt,auto outer arc]  
      $\E(tokens_n) = 1 + 2 \cdot 4^{-n}$
      \hfill
      |
      \hfill
      $\E(tokens^2_n) = 1 + 8 \cdot 4^{-n}$
      \hfill
      |
      \hfill
      $\E(tokens^3_n) = 1 + 26 \cdot 4^{-n}$
    \end{tcolorbox}
  \end{minipage}
  \caption{Herman's self stabilization algorithm with three nodes encoded as a probabilistic loop together with three moments of $tokens$.}
  \label{fig:herman3}
\end{figure}

%% file: 02-preliminaries.tex
\section{Preliminaries}\label{sec:preliminaries}

We use the symbol $\P$ for probability measures and $\E$ for the expectation operator.
The support of a random variable $X$ is denoted by $\supp(X)$.

\subsection{Probability Theory}\label{ssec:prob-theory}

Operationally, a probabilistic program is a Markov chain with potentially uncountably many states.
Let us recall some notions about Markov chains.
For more details on Markov chains and probability theory in general we refer the reader to \cite{Durrett2019}.

For a fixed set $S$, a $\sigma$-algebra is a non-empty set of subsets of $S$ closed under complementation and countable unions.

\begin{definition}[Sequence Space]
Let $(S, \mathcal{S})$ be a measurable space, that is, $S$ is a set with a $\sigma$-algebra $\mathcal{S}$.
Its \emph{sequence space} is the measurable space $(S^\omega, \mathcal{S}^\omega)$ where
$S^\omega := \{ (s_1, s_2, \dots) : s_i \in S \}$ and
$\mathcal{S}^\omega$ is the $\sigma$-algebra generated by the \emph{cylinder sets} $Cyl[B_1, \dots,B_n] := \{ \theta : \theta_i \in B_i, 1 \leq i \leq n \}$ for all prefixes $B_1, \dots, B_n \in \mathcal{S}$ and all $n \in \N$.
\end{definition}

A \emph{Markov kernel} is, on a high level, a generalization of transition probabilities between states to uncountable state spaces and is required for the definition of a \emph{Markov chain}.

\begin{definition}[Markov Chain]
Let $(S, \mathcal{S}, \P)$ be a probability space and $p : S \times \mathcal{S} \to [0, 1]$ a Markov kernel.
A stochastic process $X_n$ is a \emph{Markov chain} with Markov kernel $p$ if
\begin{equation}
    \P(X_{n+1} \in B \mid X_0=x_0, X_1=x_1, \dots, X_n=x_n) = p(X_n, B).
\end{equation}
\end{definition}

Given a measurable space $(S, \mathcal{S})$ , an \emph{initial distribution} $\mu$, a stochastic process $X_n$ and a Markov kernel $p$, \emph{Kolmogorov's Extension Theorem} says that there is a unique measure $\P$ such that $X_n$ is a Markov chain in $(S^\omega, \mathcal{S^\omega}, \P)$.

For a random variable $X$, central moments $\E((X-\E(X))^k)$ can be computed from raw moments $\E(X^k)$ and vice versa through the transformation of center:
\begin{equation}
    \E\left( (X-b)^k \right) = \E\left( ((X-a) + (a-b))^k \right) = \sum_{i=0}^k \binom{k}{i} \E\left( (X-a)^i \right) (a-b)^{k-i}.
\end{equation}

\subsection{Linear Recurrences}\label{ssec:recurrences}

We briefly recall standard terminology on algebraic sequences and recurrences. For further details, we refer the reader to \cite{Kauers2011}.
A sequence $(a_n)_{n=0}^\infty$ is called \emph{C-finite} if it obeys a linear recurrence with constant coefficients, that is, $(a_n)_{n=0}^\infty$ satisfies an equation of the form
\begin{equation*}
    a_{n+l} = c_{l-1} \cdot a_{n-l-1} + c_{l-2} \cdot a_{n-l-2} + \dots + c_0 \cdot a_n,
\end{equation*}
for some \emph{order} $l \in \N$, some constants $c_i \in \R$ and all $n \in \N$.

\begin{theorem}[Closed-form \cite{Kauers2011}]
\label{thm:closed-form}
Every C-finite sequence $(a_n)_{n=0}^\infty$ can be written as an \emph{exponential polynomial}, that is $a_n = \sum_{i=1}^m n^{d_i} u_i^n$ for some natural numbers $d_i \in \N$ and complex numbers $u_i \in \Complex$.
We refer to $\sum_{i=1}^m n^{d_i} u_i^n$ as the \emph{closed-form} or the \emph{solution} of the sequence $(a_n)_{n=0}^\infty$ or its recurrence.
\end{theorem}

An important fact is that closed-forms of linear recurrences with constant coefficients of \emph{any order} always exist and are computable.
This also holds for all variables in \emph{systems} of linear recurrences with constant coefficients.

%% file: 03-model.tex
\section{Probabilistic Program Model}\label{sec:model}

In this section we introduce our programming model (Section~\ref{sec:model:syntax}) and describe its semantics in terms of Markov chains (Section~\ref{ssec:PP-semantics}).
Moreover, we introduce transformations (\ref{ssec:PP-transforms}) normalizing a probabilistic program to simplify its analysis.

\subsection{Probabilistic Program Syntax}\label{sec:model:syntax}

\begin{figure}
    {
    \begin{minipage}{0.73\linewidth}
    \footnotesize
    $\mathit{lop} \in \{ \mathit{and}, \mathit{or} \}$,
    $\mathit{cop} \in \{ =, \neq, <, >, \geq, \leq \}$,
    $\mathit{Dist} \in \{ \mathit{Bernoulli}, \mathit{Normal}, \mathit{Uniform}, \dots \}$
    \begin{grammar}
    	<sym> ::= "a" | "b" | $\dots$ <var> ::= "x" | "y" | $\dots$
    	
    	<const> ::= $r \in \R$ | <sym> | <const> ( "+" | "*" | "/" ) <const>
    	
    	<poly> ::= <const> | <var> | <poly> ("+" | "-" | "*") <poly> | <poly>"**n"
    	
    	<assign> ::= <var> "=" <assign\_right> | <var> "," <assign> "," <assign\_right>
    	
    	<categorical> ::= <poly> ("\{"<const>"\}" <poly>)* ["\{"<const>"\}"]
    	
    	<assign\_right> ::= <categorical> | Dist"("<poly>$^*$")" | "Exponential("<const>"/"<poly>")"
    	
    	<bexpr> ::= "true" ($\star$) | "false" | <poly> <cop> <poly> | "not" <bexpr> | <bexpr> <lop> <bexpr>
    	
    	<ifstmt> ::= "if" <bexpr>":" <statems> ("else if" <bexpr>":" <statems>)$^*$ ["else:" <statems>] "end"
    	
    	<statem> ::= <assign> | <ifstmt> \quad \quad \quad <statems> ::= <statem>$^+$
    	
    	<loop> ::= <statem>* "while" <bexpr> ":" <statems> "end"
    \end{grammar}
    \end{minipage}
    }
    \caption{Grammar describing the syntax of probabilistic loops $\langle \textit{loop} \rangle$.}
    \label{fig:syntax}
\end{figure}

The syntax defining our program model is given by the grammar in Figure~\ref{fig:syntax}.
Throughout the paper, we will use the phrases \emph{(probabilistic) loops} and \emph{(probabilistic) programs} interchangeably for loops adhering to the syntax in Figure~\ref{fig:syntax}.
In our work, we infer higher moments $\E(x_n^k)$ of program variables $x$ parameterized by the loop counter $n$. 
We abstract from concrete loop guards by defining the guards of programs in our program model to be \emph{true} (written as~$\star$).
Guarded loops \lstinline{while $\phi$: $\dots$} can be modeled as an infinite loops \lstinline{while $\star$: if $\phi$: $\dots$}, with the limit behaviour giving the moments after termination (cf. Section~\ref{subsec:guarded-loops}).

Our program model defined in Figure~\ref{fig:syntax} contains non-nested while-loops which are preceded by a loop-free initialization part.
The loop-body and initialization part allow for (nested) if-statements, polynomial arithmetic, drawing from common probability distributions, and symbolic constants.
Symbolic constants can be used to represent arbitrary real numbers and are also used for uninitialized program variables. 
Categorical expressions (defined by the non-terminal $\langle \textit{categorical} \rangle$ in Figure~\ref{fig:syntax}) are expressions of the form $v_1 \{p_1\} \dots v_l \{p_l\}$ such that $\sum p_i = 1$.
Their intended meaning is that they evaluate to $v_i$ with probability $p_i$.
The last parameter $p_l$ can be omitted and in that case is set to $p_l := 1 - \sum_{i=1}^{l-1} p_i$.
For a program $\Program$ we denote with $\Vars(\Program)$ the set of $\Program$'s variables appearing on the left-hand side of an assignment in $\Program$'s loop-body.
The programs of Figures~\ref{fig:running-example}-\ref{fig:herman3} are examples of our program model defined in Figure~\ref{fig:syntax}.
In comparison to the \emph{probabilistic Guarded Command Language (pGCL)}~\cite{Barthe2020}, programs of our model contain exactly one while-loop, no non-determinism
\footnote{Non-determinism is different from probabilistic choice. Demonic (angelic) non-determinism is concerned with the worst-case (best-case) behavior. For instance, a variable can be assigned to $0$ or $1$ both with probability $\nicefrac{1}{2}$. This is different from assigning $0$ or $1$ non-deterministically, where the probability is not specified.}
but support continuous distributions and simultaneous assignments.

\subsection{Program Semantics}\label{ssec:PP-semantics}
In what follows we define the semantics of probabilistic programs in terms of Markov chains on a measurable space.
We then introduce the notion of \emph{normalized} probabilistic loops by means of so-called \emph{$\Program$-preserving program transformations} (Section~\ref{ssec:PP-transforms}).

\begin{definition}[State \& Run Space]
Let $\Program$ be a probabilistic program with $m$ variables.
We denote by \emph{$\text{ND}(\Program)$  the non-probabilistic program obtained from $\Program$} by replacing every probabilistic choice $C$ in $\Program$ by a non-deterministic choice over $\supp(C)$.
Let $\States_\Program \subseteq \R^m$ be the set of program states of $\text{ND}(\Program)$ reachable from any initial state.
The  \emph{state space} of $\Program$ is the measurable space $(\States_\Program, \mathcal{S}_\Program)$, where $\mathcal{S}_\Program$ is the Borel $\sigma$-algebra on $\R^m$ restricted to $\States_\Program$.
The  \emph{run space} of $\Program$ is the sequence space $(\States_\Program^\omega, \mathcal{S}_\Program^\omega) =: (\Runs_\Program, \mathcal{R}_\Program)$.
\end{definition}

In what follows, we omit the subscript $\Program$ whenever the program $\Program$ is irrelevant or clear from the context.
Executions/runs of a probabilistic program $\Program$ define a stochastic process, as follows. 

\begin{definition}[Run Process]\label{def:run-process}
Let $\Program$ be a probabilistic program with $m$ variables.
The \emph{run process} $\Phi_n : \Runs \to \States$ is a stochastic process in the run space mapping a program run to its $n$th state, that means, $\Phi_n(run) := run_n$.

For program variable $x$ with index $i \geq 1$, we denote by $x_n$ \emph{the projection of $\Phi_n$ to its $i$th component} $\Phi_n(\cdot)(i)$.
Given an arithmetic expression $A$ over $\Program$'s variables, we write $A_n$ for the stochastic process where every program variable $x$ in $A$ is replaced by $x_n$.
\end{definition}

\begin{remark}
Given an initial distribution of program states $\mu$ and a Markov kernel $p$ defined according to the \emph{standard meaning} of the program statements, by \emph{Kolmogorov's Extension Theorem}
we conclude that there  is a unique probability measure $\P_\Program$ on $(\Runs_\Program, \mathcal{R}_\Program)$ such that the run process is a Markov chain.
$(\Runs_\Program, \mathcal{R}_\Program, \P_\Program)$ is the probability space associated to program $\Program$.
Distributions and (higher) moments of $\Program$'s variables are to be understood with respect to this probability space.
\end{remark}

For probabilistic loops according to the syntax in Figure~\ref{fig:syntax}, the initial distribution $\mu$ of values of loop variables is the distribution of states after the statements $\langle \textit{statem} \rangle^*$ just before the while-loop.
Moreover, the loop body in Figure~\ref{fig:syntax} is considered to be atomic, meaning the Markov kernel $p$ describes the transition between full iterations in contrast to single statements.

\subsection{\texorpdfstring{$\Program$}{}-Preserving Transformations}\label{ssec:PP-transforms}

Our probabilistic programs defined by the grammar in Figure~\ref{fig:syntax} support rich arithmetic and complex probabilistic behavior/distributions. Such an expressivity of Figure~\ref{fig:syntax} comes at the cost of turning the analysis of programs defined by Figure~\ref{fig:syntax} cumbersome.  
In this section, we address this difficulty and introduce a number of program transformations that allow us to simplify our probabilistic programs to a so-called \emph{normal form} while preserving the joint distribution of program variables. Normal forms allow us to extract recursive properties from the program, which we will later use to compute moments for program variables (Section~\ref{sec:computing-moments}).

\paragraph{Schemas and Unification.}
The program transformations we introduce in this section build on the notion of \emph{schemas} and program parts.
A \emph{program part} is an empty word or any word resulting from any non-terminal of the grammar in Figure~\ref{fig:syntax}.
For our purposes, a schema $S$ is a program part with some subtrees in the program part's syntax tree being replaced by placeholder symbols $\dot{s_1},\dots,\dot{s_l}$.
A~\emph{substitution} is a finite mapping $\sigma = \{\dot{s_1} \mapsto p_1, \dots, \dot{s_l} \mapsto p_l \}$ where $p_1,\dots,p_l$ are program parts.
We denote by $S[\sigma]$ the program part resulting from $S$ by replacing every $\dot{s_i}$ by $p_i$, assuming $S[\sigma]$ is well-formed.
For two schemas $S_1$ and $S_2$ a substitution $u$ such that $S_1[u] = S_2[u]$ is called a \emph{unifier} (with respect to $S_1$ and $S_2$).
In this case $S_1$ and $S_2$ are called \emph{unifiable} (by $u$).

\paragraph{Transformations.}
In what follows, we consider $\Program$ to be a fixed probabilistic program defined by Figure~\ref{fig:syntax} and give all definitions relative to $\Program$.
A \emph{transformation} $T$ is a mapping from program parts to program parts with respect to a schema $\textit{Old}$.
$T$ is \emph{applicable} to a subprogram $S$ of $\Program$ if $S$ and $\textit{Old}$ are unifiable by the unifier $u$.
Then, the transformed subprogram is defined as $T(S) := \textit{New}[u]$ where $\textit{New}$ is a schema depending on $\textit{Old}$ and $u$.
A transformation is fully specified by defining how $\textit{New}$ results from $\textit{Old}$ and $u$.
We write $T(\Program, S)$ for the program resulting from $\Program$ by replacing the subprogram $S$ of $\Program$ by $T(S)$.

The first transformation we consider removes simultaneous assignments from $\Program$. For this, we store a copy of each assignment in an auxiliary variable to preserve the values used for simultaneous assignments, in case an assigned variable appears in an assignment expression. Variables are then assigned their intended value.

\pagebreak 
\begin{definition}[Simultaneous Assignment Transformation]
\label{def:simult-transf}
A \emph{simultaneous assignment transformation} is the transformation defined by

\lstinline{$\dot{x_1}$, $\dots$, $\dot{x_l}$ = $\dot{v_1}$, $\dots$, $\dot{v_l}$}
\ $\mapsto$\ \
\lstinline{$t_1$=$\dot{v_1}$; $\dots$; $t_l$=$\dot{v_l}$; $\dot{x_1}$=$t_1$; $\dots$; $\dot{x_l}$=$t_l$,}
\ \ 
\\where $t_1,\dots,t_l$ are fresh variables.
\end{definition}

In what follows, we assume that parameters of common distributions used in programs are constant.
Nevertheless, the following transformation enables the use of some non-constant distribution parameters.

\begin{definition}[Distribution Transformation]
\label{def:dist-transf}

A \emph{distribution transformation} is a transformation defined by either of the mappings

\begin{itemize}
    \item \lstinline{$\dot{x}$ = Normal($\dot{p}$,$\dot{v}$)} \tabto{2.75cm}$\transmap$ \lstinline{t = Normal($0$,$\dot{v}$); $\dot{x}$=$\dot{p}$+t}
    \item \lstinline{$\dot{x}$ = Uniform($\dot{p_1}$,$\dot{p_2}$)} \tabto{2.75cm}$\transmap$ \lstinline{t = Uniform($0$,$1$); $\dot{x}$=$\dot{p_1}$+($\dot{p_2}$-$\dot{p_1}$)*t}
    \item \lstinline{$\dot{x}$ = Laplace($\dot{p}$,$\dot{b}$)} \tabto{2.75cm}$\transmap$ \lstinline{t = Laplace($0$,$\dot{b}$); $\dot{x}$=$\dot{p}$+t}
    \item \lstinline{$\dot{x}$ = Exponential($\dot{c}$/$\dot{p}$)} \tabto{2.75cm}$\transmap$ \lstinline{t = Exponential($\dot{c}$); $\dot{x}$=$\dot{p}$*t}
\end{itemize}
where, for every mapping, $t$ is a fresh variable.
\end{definition}

\begin{example}
Consider Figure~\ref{fig:running-example}. The simultaneous assignment \lstinline{l,g = Laplace(x+y,$1$),Normal($0$,$1$)} can be transformed using the transformation rules from Definitions~\ref{def:simult-transf}-\ref{def:dist-transf} as follows:

\begin{minipage}{\linewidth}
    \centering
    \begin{minipage}{0.08\linewidth}
    $\overset{\scriptscriptstyle\mathrm{(sim)}}{\transmap}$
    \end{minipage}
    \begin{minipage}{0.25\linewidth}
    \begin{lstlisting}
    t1=Laplace(x+y,1)
    t2=Normal(0,1)
    x=t1
    y=t2
    \end{lstlisting}
    \end{minipage}
    \begin{minipage}{0.08\linewidth}
    $\overset{\scriptscriptstyle\mathrm{(dist)}}{\transmap}$
    \end{minipage}
    \begin{minipage}{0.35\linewidth}
    \begin{lstlisting}
    t3=Laplace(0,1)
    t1=x+y+t3
    t2=Normal(0,1)
    x=t1
    y=t2
    \end{lstlisting}
    \end{minipage}
\end{minipage}
\end{example}

To simplify the structure of probabilistic loops, we assume \lstinline{else if} branches to be syntactic sugar for nested \lstinline{if else} statements. We remove \lstinline{else} by splitting it into if-statements (\lstinline{if $C$} and \lstinline{if not $C$}). Since variables in $C$ could be changed within the first branch, we store their original values in auxiliary variables and use those for the condition $C'$ of the second \lstinline{if} statement. We capture this transformation in the following definition.

\begin{definition}[Else Transformation]
\label{def:else-transf}
An \emph{else transformation} is the transformation

\noindent
\begin{minipage}{\linewidth}
    \centering
    \begin{minipage}{0.25\linewidth}
    \begin{lstlisting}
    if $\dot{C}$: $\dddot{Branch_1}$
    else: $\dddot{Branch_2}$ end
    \end{lstlisting}
    \end{minipage}
    \begin{minipage}{0.1\linewidth}
    , u $\transmap$
    \end{minipage}
    \begin{minipage}{0.3\linewidth}
    \begin{lstlisting}
    $t_1$=$x_1$;$\dots$;$t_l$=$x_l$
    if $\dot{C}$: $\dddot{Branch_1}$ end
    if not $C'$: $\dddot{Branch_2}$ end
    \end{lstlisting}
    \end{minipage}
\end{minipage}

where $x_1,\dots,x_l$ are all variables appearing in $\dot{C}[u]$ which are also being assigned in $\dddot{Branch_1}[u]$.
Every $t_i$ is a fresh variable and
$C'$ results from $\dot{C}[u]$ by substituting every $x_i$ with $t_i$.
\end{definition}

To further simplify the loop body into a flattened list of assignments, we equip every assignment $a$ of form \lq\lq\lstinline{x = value}\rq\rq with a condition $C_a$ (initialized to true $\top$) and a default variable $d_a$ (initialized to $x$), written as \lq\lq\lstinline{x = value [$C_a$] $d_a$}\rq\rq.
The semantics of the conditioned assignment is that $x$ is assigned $value$ if $C_a$ holds just before the assignment and $d_a$ otherwise.
With conditioned assignments, the loop body's structure can be flattened using the following transformation.

\begin{definition}[If Transformation]
\label{def:if-transf}
An \emph{if transformation} is the transformation defined by

\noindent
\begin{minipage}{\linewidth}
    \centering
    \begin{minipage}{0.18\linewidth}
        \begin{lstlisting}
        if $\dot{C_1}$:
        $\dot{x}$ = $\dot{v}$ [$\dot{C_2}$] $\dot{x}$
        $\dddot{Rest}$ end
        \end{lstlisting}
    \end{minipage}
    \begin{minipage}{0.1\linewidth}
    , u $\transmap$
    \end{minipage}
    \begin{minipage}{0.3\linewidth}
        \begin{lstlisting}
        t = $\dot{x}$
        $\dot{x}$ = $\dot{v}$ [$\dot{C_1}$ and $\dot{C_2}$] $\dot{x}$
        if $C$: $\dddot{Rest}$ end
        \end{lstlisting}
    \end{minipage}
\end{minipage}
where $t$ is a fresh variable and $C$ results from $\dot{C_1}[u]$ by substituting $\dot{x}[u]$ by $t$.
If $\dddot{Rest}[u]$ is empty, the line \lstinline{if $C$: $\dddot{Rest}$ end} is omitted from the result.
If $\dot{x}[u]$ does not appear in $\dot{C_1}[u]$ the line \lstinline{t = $\dot{x}$} is dropped.
\end{definition}

\begin{example}
The following program containing nested if-statements can be flattened as follows:

\noindent
\begin{minipage}{\linewidth}
    \centering
    \begin{minipage}{0.4\linewidth}
        \begin{lstlisting}
        if x == 1:
        x = Bernoulli(1/2)
        if x == 0: y = 1 end
        end
        \end{lstlisting}
    \end{minipage}
    \begin{minipage}{0.07\linewidth}
    $\overset{\scriptscriptstyle\mathrm{(if)}}{\transmap}$
    \end{minipage}
    \begin{minipage}{0.33\linewidth}
        \begin{lstlisting}
        if x == 1:
        x = Bernoulli(1/2)
        y = 1 [x == 0] y
        end
        \end{lstlisting}
    \end{minipage}
    \begin{minipage}{0.07\linewidth}
    $\overset{\scriptscriptstyle\mathrm{(if)}}{\transmap}$
    \end{minipage}
    
    \begin{minipage}{0.4\linewidth}
        \begin{lstlisting}
        t = x
        x = Bernoulli(1/2) [t == 1] x
        if t == 1:
        y = 1 [x == 0] y
        end
        \end{lstlisting}
    \end{minipage}
    \begin{minipage}{0.07\linewidth}
    $\overset{\scriptscriptstyle\mathrm{(if)}}{\transmap}$
    \end{minipage}
    \begin{minipage}{0.4\linewidth}
        \begin{lstlisting}
        t = x
        x = Bernoulli(1/2) [t == 1] x
        y = 1 [t == 1 $\land$ x == 0] y
        \end{lstlisting}
    \end{minipage}
\end{minipage}
\end{example}

To bring further simplicity to our program $\Program$, we ensure for each variable to be modified only once within the loop body. To remove duplicate assignments we introduce new variables $x_1, \dots, x_{l-1}$ to store intermediate states.
Assignments to other variables, in between the updates of $x$, will be adjusted to refer to the latest $x_i$ instead of $x$.

\begin{definition}[Multi-Assignment Transformation]
\label{def:multi-transf}
A \emph{multi-assignment transformation} is the transformation defined by

\noindent
\begin{minipage}{\linewidth}
    \centering
    \begin{minipage}{0.28\linewidth}
    \begin{lstlisting}
    $\dot{x}$ = $\dot{v_1}$ [$\dot{C_1}$] $\dot{x}$;$\dddot{Rest_1}$;
    $\dot{x}$ = $\dot{v_2}$ [$\dot{C_2}$] $\dot{x}$;$\dddot{Rest_2}$;
    $\dots$; $\dot{x}$ = $\dot{v_l}$ [$\dot{C_l}$] $\dot{x}$;
    \end{lstlisting}
    \end{minipage}
    \begin{minipage}{0.1\linewidth}
    , u $\transmap$
    \end{minipage}
    \begin{minipage}{0.3\linewidth}
    \begin{lstlisting}
    $x_1$ = $\dot{v_1}$ [$\dot{C_1}$] $\dot{x}$;$Rest_1$;
    $x_2$ = $v_2$ [$C_2$] $x_1$;$Rest_2$;
    $\dots$; $\dot{x}$ = $v_l$ [$C_l$] $x_{l-1}$;
    \end{lstlisting}
    \end{minipage}
\end{minipage}
where $x_1,\dots,x_{l-1}$ are fresh variables.
For $i \geq 2$, $v_i$, $C_i$ and $Rest_i$ result from $\dot{v_i}[u]$, $\dot{C_i}[u]$ and $\dddot{Rest}_i[u]$, respectively, by replacing $\dot{x}[u]$ by $x_{i-1}$.
\end{definition}

\begin{example}
In the program of Figure~\ref{fig:herman3},
program line \lstinline{if x1 == x3: t1 = 1 else: t1 = 0} 
can be transformed using transformation rules from Definitions~\ref{def:else-transf}-\ref{def:multi-transf} as follows:

\noindent
\begin{minipage}{\linewidth}
    \centering
    \begin{minipage}{0.1\linewidth}
    $\overset{\scriptscriptstyle\mathrm{(else)}}{\transmap}$
    \end{minipage}
    \begin{minipage}{0.35\linewidth}
    \begin{lstlisting}
    if x1 == x3: t1 = 1
    if x1 != x3: t1 = 0
    \end{lstlisting}
    \end{minipage}
    \ \\
    \begin{minipage}{0.1\linewidth}
    $\overset{\scriptscriptstyle\mathrm{(if)}}{\transmap}$
    \end{minipage}
    \begin{minipage}{0.35\linewidth}
    \begin{lstlisting}
    t1 = 1 [x1 == x3] t1
    t1 = 0 [x1 != x3] t1
    \end{lstlisting}
    \end{minipage}
    \ \\
    \begin{minipage}{0.1\linewidth}
    $\overset{\scriptscriptstyle\mathrm{(multi)}}{\transmap}$
    \end{minipage}
    \begin{minipage}{0.35\linewidth}
    \begin{lstlisting}
    t11 = 1 [x1 == x3] t1
    t1 = 0 [x1 != x3] t11
    \end{lstlisting}
    \end{minipage}
\end{minipage}
\end{example}

With program transformations defined, we can turn our attention to program properties. In particular, we show that our transformations of $\Program$ do not change the joint distribution of $\Program$'s variables. Since our transformations may introduce new variables, we consider  program equivalence with respect to program variables in order to ensure that the distribution of $\Program$ is maintained/preserved by our transformations.

\begin{definition}[Program Equivalence]
Let $\Program_1$ and $\Program_2$ be two probabilistic programs. We define 
$\Program_1$ and $\Program_2$ to be \emph{equivalent with respect to a set of program variables $X$}, in symbols $\Program_1 \equiv^X \Program_2$, if: 
\begin{enumerate}
    \item $X \subseteq \Vars(\Program_1) \cap \Vars(\Program_2)$,  and
    \item the joint distributions of $X$ arising from $\Program_1$ and $\Program_2$ are equal.
\end{enumerate}
\end{definition}

To relate a program $\Program$ to its transformed version $T(\Program, S)$ we consider the distribution of variables of the original program $\Program$. 
If $\Program$ retains the joint distribution of its variables after applying transformation $T$, we say that $T$ is $\Program$-preserving.

\begin{definition}[$\Program$-Preserving Transformation]
We say that a transformation $T$ is \emph{$\Program$-preserving} if
$\Program \equiv^{\Vars(\Program)} T(\Program, S)$ for all subprograms $S$ of $\Program$ which are unifiable with $Old$.
\end{definition}

It is not hard to argue that the transformations defined above are $\Program$-preserving, yielding the following result.

\begin{lemma}\label{thm:ProgramPreserving}
The transformations from Definitions~\ref{def:simult-transf}-\ref{def:multi-transf} are $\Program$-preserving.
\end{lemma}

By exhaustively  applying the $\Program$-preserving transformations of Definitions~\ref{def:simult-transf}-\ref{def:multi-transf} over $\Program$, we obtain a so-called normalized program $\Program_N$, as defined below. The normalized $\Program_N$ will then further be used in computing higher moments of $\Program$ in Sections~\ref{sec:computing-moments}, as the $\Program_N$ preserves the moments of $\Program$ (Theorem~\ref{thm:NormalForm}).  

\begin{definition}[Normal Form]\label{def:normal-form}
A program $\Program$ is in \emph{normal form} or a \emph{normalized} program if none of the transformations from  Definitions~\ref{def:simult-transf}-\ref{def:multi-transf} are applicable to $\Program$.
\end{definition}

\begin{theorem}[Normal Form]\label{thm:NormalForm}
For every probabilistic program $\Program$ there is a $\Program_{\mathcal{N}}$ in normal form such that $\Program \equiv^{\Vars(\Program)} \Program_{\mathcal{N}}$.
Moreover, $\Program_{\mathcal{N}}$ can be effectively computed from $\Program$ by exhaustively applying transformations from Definitions~\ref{def:simult-transf}-\ref{def:multi-transf}.
\end{theorem}
\begin{proof}
There are two claims in the theorem, which we need to address: (i) exhaustively applying transformation rules terminates (\emph{termination}), and (ii) it preserves statistical properties of the (original) program variables (\emph{correctness}). 

For termination, we show that programs become smaller, in some sense, after every transformation. In particular, we consider the program size to be given by a tuple ($\textit{Sim}$, $\textit{Dist}$, $\textit{Else}$, $\textit{If}$, $\textit{Multi}_B$, $\textit{Multi}_I$), representing the number of simultaneous assignments, non-trivial distributions, else statements, assignments within if branches (weighted for nested ifs), and number of variables with multiple assignments in the loop body and initialization part, respectively. With respect to the lexicographic order, each transformation reduces the size of the program which is lower-bounded by~$0$.

Correctness can be shown by treating each transformation separately and showing that it does not alter the variables' distributions after a single application (Lemma~\ref{thm:ProgramPreserving}).
This is true for all transformations from Definitions~\ref{def:simult-transf}-\ref{def:multi-transf}.
Auxiliary variables are used to store the original value to prevent intervening variable modifications.
For \nameref*{def:multi-transf} (Definition~\ref{def:multi-transf}), we also revise the rest of the assignments to reflect the change of the original variable. The \nameref*{def:dist-transf} (Definition~\ref{def:dist-transf}) uses statistical properties of well-known distributions.
\end{proof}

\paragraph{Properties of Normalized Programs.}
Figure~\ref{fig:running-normalized} shows a normal form for the program from Figure~\ref{fig:running-example}.
Normalized programs have the following important properties: (1) all distribution parameters are constant; (2) the loop body is a sequence of guarded assignments; (3) every program variable is only assigned once in the loop body.
Moreover for every guarded assignment \lstinline{v = $\textit{assign}_{\textit{true}}$ [C] $\textit{assign}_{\textit{false}}$}, the guard $C$ is a boolean condition and $\textit{assign}_{\textit{false}}$ is a single variable which is assigned to \texttt{v} if $C$ evaluates to \emph{false}. If $C$ evaluates to \emph{true}, the variable \texttt{v} is assigned $\textit{assign}_{\textit{true}}$.
The expression $\textit{assign}_{\textit{true}}$ is either a distribution with constant parameters or a probabilistic choice of polynomials as illustrated in Figure~\ref{fig:running-normalized}.

\begin{remark}
Based on the order in which transformations are applied to a program $\Program$ and the names used for auxiliary variables, several different normalized programs can be achieved for $\Program$.
In this work, only the existence of a normal form is relevant.
Moreover, from the definitions of the transformation, it is apparent that exhaustively applying them leads to a normal form whose size is linear in the size of the original program.
\end{remark}

\begin{figure}
    \centering
    \noindent
    \begin{minipage}{0.5\linewidth}
        \begin{lstlisting}[basicstyle=\footnotesize]
            toggle = 0; sum = s0
            x = 1; y = 1; z = 1
            while $\star$:
              toggle = 1-toggle
              x = 1+x {1/2} 2+x [toggle==0] x
              y = y+z+x**2 {1/3} -x+y-z [toggle==0] y
              z = z+y {1/4} z-y [toggle==0] z
              t1 = Laplace(0, 1)
              l = t1+x+y
              g = Normal(0, 1)
              sum = sum+x [g < 1/2] sum
            end
        \end{lstlisting}
    \end{minipage}
    \caption{A normal form for the program in Figure~\ref{fig:running-example}.}
    \label{fig:running-normalized}
\end{figure}

%% file: 04-types.tex
\section{Finite Types in Probabilistic Programs}\label{sec:types}

Given a probabilistic program $\Program$ in our programming model, 
the transformations of Section~\ref{ssec:PP-transforms} simplify $\Program$ by computing its normalized form while maintaining the distribution (and hence also moments) of $\Program$. Nevertheless, the normalized form of $\Program$ contains computationally expensive polynomial arithmetic, potentially hindering the automated analysis of $\Program$ in Section~\ref{sec:computing-moments} due to a computational blowup.
Therefore, we introduce further simplifications for $\Program$ by means of power reduction techniques.

\begin{example} 
Consider Figure~\ref{fig:running-example} and assume we are interested in the $k$th power of variable $\textit{toggle}$ and deriving the raw moment $\E(\textit{toggle}_n^k)$. 
Our analysis relies on replacing variables with their assignments (see Section~\ref{sec:computing-moments}), leading to the expression $(1 - \textit{toggle}_{n-1})^k$. When expanded, this is a polynomial in $\textit{toggle}_{n-1}$ with $k{+}1$ monomials:
\begin{equation*}
    (1 - \textit{toggle}_{n-1})^k = \sum_{i=0}^{k} \binom{k}{i} (-1)^i \textit{toggle}_{n-1}^{i}
\end{equation*}
Higher moments, together with the aforementioned replacements, may lead to blowups of the number of monomials to consider. However, observing that the variable $\textit{toggle}$ is binary, we have $\textit{toggle}_n^k = \textit{toggle}_n = 1 - \textit{toggle}_{n-1}$ for any $k \geq 0$. 
Arbitrary powers of the finite variable $\textit{toggle}$ with~$2$ possible values can be written in terms of powers smaller than~$2$.
In the rest of this section, we show that this phenomenon generalizes from binary variables to arbitrary finite valued variables, thus simplifying the analysis of higher moments of finite valued program variables. 
\end{example}

As defined in Definition~\ref{def:run-process}, for an arithmetic expression $X$ over program variables, $X_n$ denotes the stochastic process mapping a program run to the value of $X$ after iteration $n$.

\begin{definition}[Finite Expression]
Let $\Program$ be a probabilistic program and $X$ an arithmetic expression over the variables of $\Program$. We say that 
$X$ is \emph{finite} if there exist $a_1,\dots,a_m \in \R$ such that for all $n \in \N$ : $X_n \in \{ a_1,\dots, a_m \}$.
\end{definition}

\subsection{Power Reduction for Finite Types}\label{ssec:power-reduction}

As established in~\cite[Lemma 1]{Bartocci2020}, high powers $k$ of a random variable $X$ over a finite set can be reduced.
We adapt their result to our setting as follows.

\begin{theorem}[Finite Power Reduction]\label{thm:PowerReduction}
Let $m,k\in\Z$ and $X$~be a discrete random variable over $A=\{a_1,\dots,a_m\}$. 
Then we can rewrite $X^k$ as a linear combination of $1, X, X^2, \cdots, X^{m-1}$. Furthermore, 
\begin{equation}\label{eq:closed-form}
X^k =  \overline {a^k} M^{-1} \overline X,
\end{equation}
where 
$\overline {a^k} = (a_1^k, \dots, a_m^k)$,
$M$ is an $m\times m$ matrix with $M_{ij} = a_j^{i-1}$ (with $0^0:=1$), and
$\overline X = (X^0, \dots, X^{m-1})^T$. 
\end{theorem}

In other words, Theorem~\ref{thm:PowerReduction} implies that any higher moment of $X$, can be computed from just its first $m{-}1$ moments.
Furthermore, we build on Theorem~\ref{thm:PowerReduction} and establish the inverse of matrix $M$ explicitly ($M$ is explicit in Theorem ~\ref{thm:PowerReduction} but its inverse is implicit).

\begin{theorem}[Reduction Formula]\label{thm:reduction-formula}
Recall that the $k$th elementary symmetric polynomial with respect to a set $V=\{v_1,\dots,v_n\}$ is $\displaystyle e_k(V) = \sum_{1\le j_1<\cdots <j_k\le n} v_{j_1}\cdots v_{j_k}$ and let $A_{-j} = A\setminus \{a_j\})$. Then the inverse of~$M$ in~\eqref{eq:closed-form} is given by
\begin{equation}\label{inverseFormula}
M^{-1}_{ij} = - \frac{(-1)^{j} e_{m-j}(A_{-i})}{\prod_{a\in A_{-i}}(a-a_i)}.
\end{equation}
\end{theorem}
\begin{proof}
Let $MN=B$ for $M$ as of \eqref{eq:closed-form} and $N$ as of \eqref{inverseFormula}. We show that $B=I$ by showing that $B_{ij}=1$ if $i=j$ and $B_{ij}=0$ otherwise.
We have 
\begin{align*}
\begin{split}
B_{ij} 
= \sum_{1\le k\le m} N_{ik} M_{kj} 
&= \sum_{1\le k\le m}  - \frac{(-1)^{k} e_{m-k}(A_{-i})}{\prod_{a\in A_{-i}}(a-a_i)}  a_j^{k-1} \\
&= \frac{1}{\prod_{a\in A_{-i}}(a-a_i)} \sum_{1\le k\le m} - (-1)^{k} a_j^{k-1} e_{m-k}(A_{-i}) \\
&= \frac{1}{\prod_{a\in A_{-i}}(a-a_i)} \sum_{0\le k\le m-1} (-a_j)^{k} e_{m-k-1}(A_{-i})  \\
&= \frac{1}{\prod_{a\in A_{-i}}(a-a_i)} \prod_{a\in A_{-i}}(a-a_j),
\end{split}
\end{align*}
where the last equation comes from the expansion of product $\prod_{a\in A_{-i}}(a-a_j)$ and grouping by the exponent of $a_j$. We can clearly see that the last expression is $1$ if $i=j$ and $0$ otherwise.
\end{proof}

\begin{example}
Let $X$ be a random variable over $A := \{-2,0,1,3\}$.
Using Theorem~\ref{thm:PowerReduction}-\ref{thm:reduction-formula}, we obtain the $10$th power of $X$ as: 
\begin{flalign*}
\begin{split}
    X^{10} &= \begin{pmatrix}(-2)^{10} & 0^{10} & 1^{10} & 3^{10}\end{pmatrix}
    \begin{pmatrix}
    0 & \nicefrac{-1}{10} & \nicefrac{2}{5} & \nicefrac{-1}{30} \\
    1 & \nicefrac{-5}{6} & \nicefrac{-1}{3} & \nicefrac{1}{6} \\
    0 & 1 & \nicefrac{1}{6} & \nicefrac{-1}{6} \\
    0 & \nicefrac{-1}{15} & \nicefrac{1}{30} & \nicefrac{1}{30}
    \end{pmatrix}
    \begin{pmatrix} X^0 \\ X^1 \\ X^2 \\ X^3 \end{pmatrix}
    \\
    &= 1934 X^3 + 2105 X^2 - 4038 X.
\end{split}
\end{flalign*}
\end{example}

%% file: 05-computing-moments.tex
\section{Computing Higher  Moments of Probabilistic Programs}\label{sec:computing-moments}

We now bring together the results from Sections~\ref{sec:model}-\ref{sec:types} to develop the \emph{theory of moment-computability} for probabilistic loops.
We establish the technical details leading to sufficient conditions that ensure moment-computability, culminating in the proof of Theorem~\ref{thm:moment-computability}.
The main ideas of our method are illustrated on the probabilistic loop from Figure~\ref{fig:running-example} in Example~\ref{ex:recurrences} at the end of this section.

\begin{definition}[Moment-Computability]\label{moment-computability}
A probabilistic loop $\Program$ is \emph{moment-computable} if a closed-form (according to Theorem~\ref{thm:closed-form}) of $\E(x_n^k)$ exists and is computable for all $x \in \Vars(\Program)$ and $k \in \N$.
\end{definition}

We will describe the class of moment-computable probabilistic loops through the properties of the dependencies between program variables.

\begin{definition}[Variable Dependency]\label{def:dependency}
	Let $\Program$ be a probabilistic loop and $x,y \in \Vars(\Program)$.
	We define:
	\begin{itemize}
	    \item $y$ \emph{depends conditionally on} $x$, if there is an assignment of $y$ within an if-else-statement and $x$ appears in the if-condition.
	    
	    \item $y$ \emph{depends finitely on} $x$, if $x$ is finite and appears in an assignment of $y$.
	    
	    \item $y$ \emph{depends linearly on} $x$, if $x$ appears only linearly in every assignment of $y$.
	    
		\item $y$ \emph{depends polynomially on} $x$, if there is an assignment of $y$ in which $x$ appears non-linearly and $x$ is not finite (motivated by Section~\ref{ssec:power-reduction}).
		
		\item $y$ \emph{depends on} $x$ if it depends on $x$ conditionally, finitely, linearly, or polynomially.
	\end{itemize}
	Furthermore, we consider the transitive closure for variable dependency as follows:
	If $z$ depends on $y$ and $y$ depends on $x$, then $z$ depends on $x$.
	If one of the two dependencies is polynomial, then $z$ depends polynomially on $x$.
\end{definition}

A crucial point to highlight in Definition~\ref{def:dependency} is that due to transitivity, variables can depend on themselves.
For instance, if variable $x$ depends on $y$ and $y$ on $x$, then $x$ is self-dependent.
Moreover, if either of the dependencies between $x$ and $y$ is non-linear, $x$ depends, by Definition~\ref{def:dependency}, \emph{polynomially on itself.}
The absence of such polynomial self-dependencies is a central condition for our notion of moment-computable loops.

\begin{theorem}[Moment-Computability]\label{thm:moment-computability}
    A probabilistic loop $\Program$ is moment-computable if \emph{(1) none of its non-finite variables depends on itself polynomially, and (2) if the variables in all if-conditions are finite.}
\end{theorem}

Note that none of the program transformations from Section~\ref{ssec:PP-transforms} can introduce a polynomial (self-)dependence.
We capture this in the following lemma:
\begin{lemma}[Non-Dependency Preservation]\label{lemma:simplicitypreservation}
    If a variable $x \in \Vars(\Program)$ does not depend on itself polynomially, neither does $x \in \Vars(\Program_{\mathcal{N}})$.
\end{lemma}

Before we prove  Theorem~\ref{thm:moment-computability}, let us first show its validity  for programs in normal form (as defined in Section~\ref{ssec:PP-transforms}).
Recall that a normalized program's loop body is a flat list of (guarded) assignments, one for every (possibly auxiliary) program variable.

\begin{lemma}[Normal Moment-Computability]\label{lemma:moment-computability}
    The Moment-Computability Theorem (Theorem~\ref{thm:moment-computability}) holds for loops in normal form.
\end{lemma}
\begin{proof}
We have to show that for an arbitrary normalized program $\Program$ satisfying the conditions of Theorem~\ref{thm:moment-computability}, all $x \in \Vars(\Program)$ and all $k \in \N$, the $k$th moment of $x$ (that is $\E(x_n^k)$) admits a closed-form as an exponential polynomial.
$\E(x_n^k)$ admits a closed-form as an exponential polynomial in $n$ if it satisfies a linear recurrence.
We show a slightly more general statement.
That is, we show that for \emph{any monomial of program variables} $M$ (and hence also for $x^k$) $\E(M)$ satisfies a linear recurrence.
The main idea of the proof is to show that $\E(M)$ only depends on a finite set of monomials, each (in some sense) \emph{not larger} than~$M$ itself.
Intuitively, the finite set of monomials on which $\E(M)$ depends on are all monomials of program variables such that their expected values determine $\E(M)$.
We will show that this set of monomials exists, is finite, and leads to a system of linear recurrences containing $\E(M)$, implying a computable exponential polynomial closed-form for $\E(M)$ by Theorem~\ref{thm:closed-form}.

Let $\Program$ be a normalized program satisfying the conditions of Theorem~\ref{thm:moment-computability}, $x \in \Vars(\Program)$, $k,n\in\N$ arbitrary, and $\mathcal{M}$ the set of all monomials over $\Vars(\Program)$ with the powers of every finite variable $d$ bounded by the number of possible values of $d$ (higher powers can be reduced as of Theorem~\ref{thm:PowerReduction}).

\paragraph{Recurrences over Moments.}
Given the syntax of probabilistic programs and properties of expectation, for any monomial $M \in\mathcal{M}$ there is a natural way to express the expected value of $M$ in iteration $n{+}1$, that is $\E(M_{n+1})$, as a linear combination of expectations of monomials in iteration $n$:
\begin{equation}\label{eq:monomial-recurrence}
    \E(M_{n+1}) = \sum_{N\in M^*} c_N \E(N_n),
\end{equation}
for some finite set $M^*\subset\mathcal{M}$, and non-zero constants~$c_N$.
Equation \eqref{eq:monomial-recurrence} is called \emph{the recurrence} of $\E(M)$.
The set $M^*$ is the set of monomials that appear in the recurrence of $\E(M)$.
We define the $^*$ operator to give such a set for any monomial and extend the definition to sets by 
\begin{equation*}
    S^* = \bigcup_{M\in S} M^*.
\end{equation*}

The exact recurrence can be computed from $\E(M_{n+1})$ by replacing variables appearing in $M$ by their assignments and using the linearity of $\E$ to convert expected values of polynomials to linear combinations of expected monomials.
Recall that for a program in normal form, every program variable is only assigned once, all distribution parameters are constant, and the loop body is a flat list of guarded assignments (Section~\ref{ssec:PP-transforms}).
The guarded assignments are of the form
$$x = a_0 \{p_0\} \dots \{p_{i-1}\} a_i [C_x] d_x,$$
for polynomials of program variables $a_0, \dots, a_i$ and constant probabilities $p_0,\dots,p_{i-1}$, or
$$x = Dist [C_x] d_x$$ for some admissible distribution $Dist$.
The guard $[C_x]$ is a boolean condition and for normalized programs, $d_x$ is always a program variable.
Assume, that the variable $x$ appears in the monomial $M$.
Hence, $M = M'\cdot x^k_{n+1}$ for some monomial $M'$ not containing $x$.
If the single assignment of $x$ is a (guarded) probabilistic choice of polynomials we rewrite $\E(M_{n+1})$ to
\begin{equation}\label{eq:rewrite-cat}
    \E(M_{n+1}) = \E(M'\cdot x^k_{n+1}) = \E\left(M'\left( d_x [\lnot C_x] + \sum p_i a_i^k [C_x]\right)\right).
\end{equation}

If the single assignment of $x$ is a (guarded) draw from a distribution we rewrite $\E(M_{n+1})$ to
\begin{equation}\label{eq:rewrite-dist}
    \E(M_{n+1}) = \E(M'\cdot x^k_{n+1}) = \E\left(M' d_x [\lnot C_x]\right) + \E\left(M'[C_x]\right)\E\left(Dist^k\right).
\end{equation}

Variables in conditions $[C_x]$ are all finite since they come from branch conditions. 
We further simplify the expressions of Equations~\eqref{eq:rewrite-cat}-\eqref{eq:rewrite-dist} by replacing the logical conditions $[C_x]$ by polynomials that evaluate to $1$ whenever variables satisfy the condition $[C_x]$ and to $0$ otherwise.
It is possible to write any logical condition over finitely valued variables as such a polynomial (\cite{Stankovic2022}),
with $[x=c] := \prod_{d\in\omega(x)\setminus\{c\}} \frac{x-d}{c-d}$, $[\lnot C] = 1 - [C]$, and $[C_1 \land C_2] = [C_1]\cdot[C_2]$, where $\omega(x)$ is the set of possible values of $x$\footnote{Because negation and conjunction are functionally complete for propositional logic, we can also handle disjunctions using De Morgan’s laws: $[P \lor Q] = [\lnot (\lnot P \land \lnot Q)] = 1 - (1 - [P])\cdot(1 - [Q])$. Inequalities can then be transformed into a disjunction of equalities since we assume that all variables appearing in if-conditions are finitely valued.}.
Converting conditions to polynomials in the equation above leads to polynomials over moments of program variables.
We can compute the recurrence of $\E(M)$ in Equation~\eqref{eq:monomial-recurrence} by replacing all variables in $M$ of iteration $n{+}1$ from last to first (by their appearance in $\mathcal{P}$'s loop body).
Throughout, the linearity of $\E$ is used to convert expected values of polynomials to linear combinations of expected monomials.
In Example~\ref{ex:recurrences}, we illustrate this computation on the program from Figure~\ref{fig:running-example} with its normal form from Figure~\ref{fig:running-normalized}.

\paragraph{Ordering.}
Now that we can compute the recurrences, we need to introduce the order for monomials, such that the monomials appearing in the recurrence for $M$ are not larger than $M$. We will need this to show that computing the recurrences as described above is, indeed, a finite process. Intuitively, a variable $y$ is larger than (or equal to) $x$ if it depends on $x$. Mutually dependent variables will form an equivalence class. We then extend the order to monomials based on their degrees with respect to the variables' equivalence classes. 

More formally, let $\preceq$ be a smallest total preorder on variables such that $x\preceq y$ whenever $y$ depends on~$x$. We write $x\prec y$ iff $x \preceq y$ and $y \not\preceq x$, and $x\sim y$ iff $x \preceq y$ and $y \preceq x$. Let $\widetilde{x}$ be the equivalence class of $x$ induced by $\sim$.
Note, that all variables in an equivalence class are mutually dependent.
However, because of our assumptions on the program $\Program$, the mutual dependencies among non-finite variables are all linear (as there are no polynomial self-dependencies).

We extend $\preceq$ to a preorder on the set of monomials $\mathcal{M}$. For every monomial $M$ and non-finite variable~$x$, we consider the degree $deg(\widetilde{x}, M)$ of $M$ in the equivalence class of $\widetilde{x}$\footnote{$deg(\widetilde{x}, M) := \sum_{x\in\widetilde{x}} deg(x, M)$, where $deg(x, M)$ is the degree of $x$ in $M$.}. We associate $M$ with the sequence of $deg(\widetilde{x}, M)$ for equivalence classes of all non-finite variables, ordered reverse-lexicographically with respect to $\preceq$.
Then the relations $\prec$, $\sim$, and equivalence classes $\widetilde{(-)}$ follow naturally from $\preceq$. 

By Theorem~\ref{thm:PowerReduction} and the definition of~$\preceq$, the equivalence class $\widetilde{M}$ is finite for each $M\in\mathcal{M}$. Let $\mathcal{M}_\sim$ be the set of equivalence classes of $\sim$. The preorder $\preceq$ induces a partial order $\preceq_\sim$ on $\mathcal{M}_{\sim}$.
Note that monomials only contain non-negative powers and a finite number of variables.
These facts together with $\preceq$ being total imply that $\preceq_\sim$ is a well-order. We will write $\preceq$ instead of $\preceq_\sim$ when the meaning is clear from the context.

With these orders defined, we have $N \preceq M$ for any $N\in M^*$ and $M\in\mathcal{M}$.
Intuitively, this means that for every monomial $M$, the monomials occurring in the recurrence of $M$ are not larger than $M$ itself.
This is true because the order on variables is defined according to variable dependencies, the order on monomials is a reverse-lexicographic extension of the order on variables, and the fact that polynomial self-dependencies are not allowed by assumption.

\paragraph{Finite $\mathcal{S}^{(Q)}$}
We show, that for any monomial $Q$, there is a finite set of monomials $\mathcal{S}^{(Q)}\subset\mathcal{M}$ containing $Q$ such that $M^*\subset \mathcal{S}^{(Q)}$ for any $M\in\mathcal{S}^{(Q)}$.
Because $M^*$ is the set of all monomials in the recurrence of $M$, this means that $\mathcal{S}^{(Q)}$ contains all monomials necessary to construct a system of linear recurrences containing $E(Q)$ (if it exists and is finite).

Let 
\begin{equation}\label{eq:SQ}
  \mathcal{S}^{(Q)} := \widetilde{Q} \cup \bigcup_{\substack{A\in\widetilde{Q^{}}^*\\A\prec Q}} \mathcal{S}^{(A)}.
\end{equation}
Clearly $Q\in\mathcal{S}^{(Q)}$ and $M^*\subset \mathcal{S}^{(Q)}$ for any $M\in\mathcal{S}^{(Q)}$ by construction. We are left to show that $\mathcal{S}^{(Q)}$ is finite for all $Q \in \widetilde{Q} \in \mathcal{M}_\sim$, which we can do by transfinite induction over $\mathcal{M}_\sim$. 

For the base case, we have to show that $\mathcal{S}^{(Q)}$ is finite for all $Q$ in $\widetilde{1}$ (the trivial monomial). Since $\widetilde{1} = \{\prod_{d\in F} d^{\lambda_d} \mid \lambda_d\le m_d\}$, where $F$ is the set of finite program variables and $m_d$ are upper bounds on their powers as of Theorem~\ref{thm:PowerReduction}, 
$\mathcal{S}^{(Q)} = \widetilde{1}$ is finite for all $Q\in \widetilde{1}$.
Suppose $\mathcal{S}^{(A)}$ is finite for all $A\prec Q$. Since $\widetilde{Q^{}}^*$ is finite, so is the union in~\eqref{eq:SQ} and, as a result,~$\mathcal{S}^{(Q)}$.

\paragraph{Moments.} A system of linear recurrences with constant coefficients can be constructed for monomials in $\mathcal{S}^{(Q)}$.
Therefore, the closed-form of any $E(M) \in \mathcal{M}$ exists and is computable.
\end{proof}

We now turn back to Theorem~\ref{thm:moment-computability} and establish its validity. The crux of our proof below relies on the fact that  our transformations 
computing normal forms of $\Program$ (see Section~\ref{ssec:PP-transforms}) are $\Program$-preserving. 
\begin{proof}[Proof (of Theorem~\ref{thm:moment-computability})]
    By Theorem~\ref{thm:NormalForm} (\nameref*{thm:NormalForm} Termination), $\Program$ can be transformed to a~normalized loop $\Program_{\mathcal{N}}$.
    By Lemma~\ref{lemma:simplicitypreservation} (\nameref*{lemma:simplicitypreservation}), $\Program_{\mathcal{N}}$ satisfies all conditions from Lemma~\ref{lemma:moment-computability} (\nameref*{lemma:moment-computability}).
    Thus, $\Program_{\mathcal{N}}$ is moment-computable and the moments are equivalent to those of $\Program$ for $x \in \Vars(\Program)$ by Theorem~\ref{thm:NormalForm} (\nameref*{thm:NormalForm} Correctness).
\end{proof}

The proofs of Theorem~\ref{thm:moment-computability} and Lemma~\ref{lemma:moment-computability} are constructive and describe a procedure to compute (higher) moments of program variables by (1) transforming a probabilistic loop into a normal form according to Theorem~\ref{thm:NormalForm}, (2) constructing a system of linear recurrences as in the proof of Lemma~\ref{lemma:moment-computability} and (3) solving the system of linear recurrences with constant coefficients.
In the following example, we illustrate the whole procedure on the running example from Figure~\ref{fig:running-example}.

\begin{example}\label{ex:recurrences}
We return to the probabilistic loop $\Program$ from Figure~\ref{fig:running-example}.
A normal form $\Program_{\mathcal{N}}$ for $\Program$ was given in Figure~\ref{fig:running-normalized}.
To compute a closed-form of the expected value of the program variable $z$, we will model $\E(z_n)$ as a system of linear recurrences according to the proof of Lemma~\ref{lemma:moment-computability}.
For a cleaner presentation we will refer to the variable \emph{toggle} by $t$.
Note that $t$ is binary.
To construct the recurrence for $\E(z_n)$, we start with $\E(z_{n+1})$ at the assignment of $z$ in $\Program_{\mathcal{N}}$ and repeatedly replace variables by the right-hand side of their assignment starting from $z$'s assignment and stopping at the top of the loop body.
Throughout the process we use the linearity of expectation to convert expected values of polynomials to linear combinations of expected monomials (as required by Equation~\ref{eq:monomial-recurrence}):
{
\footnotesize
\begingroup
\addtolength{\jot}{0.5em}
\begin{flalign*}
&\E(z_{n+1}) \\
&\quad \downarrow \text{assignment of $z$} \\
&\E(([t_{n+1}=0](z_n + y_{n+1})\ \{\nicefrac{1}{4}\}\ [t_{n+1}=0](z_n - y_{n+1})) + [t_{n+1} \neq 0]z_n) \\
&\quad \downarrow \text{replace (in)equalities by polynomials} \\
&\E(((1-t_{n+1})(z_n + y_{n+1})\ \{\nicefrac{1}{4}\}\ (1-t_{n+1})(z_n - y_{n+1})) + t_{n+1}z_n) \\
&\quad \downarrow \text{$\E$ on categorical} \\
&\frac{1}{4}\E((1-t_{n+1})(z_n + y_{n+1})) + \frac{3}{4}\E((1-t_{n+1})(z_n - y_{n+1})) + \E(t_{n+1}z_n) \\
&\quad \downarrow \text{simplify; $\E$ linearity} \\
&\E(z_n) + \frac{1}{2}\E(t_{n+1}y_{n+1}) - \frac{1}{2}\E(y_{n+1})  \\
&\quad \downarrow_{*} \text{similarly, replace $y_{n+1},x_{n+1},t_{n+1}$} \\
&\E(z_n) - \frac{1}{6} \E(t_n x_n) - \frac{1}{2} \E(t_n y_n) - \frac{1}{6} \E(t_n x_n^2) + \frac{1}{12} \E(t_n) + \frac{1}{6} \E(t_n z_n).
\end{flalign*}
\endgroup
}
The last line of the calculation represents the recurrence equation of the expected value of $z$.
For every monomial in the recurrence equation of $\E(z_n)$ (that is: $z^* = $ $\{tx, ty, tx^2, t, tz\}$), we compute its respective recurrence equation and recursively repeat this procedure which eventually terminates according to Lemma~\ref{lemma:moment-computability}.

The ordering of variables and monomials, which is essential for proving termination in Lemma~\ref{lemma:moment-computability} comes from an ordering of variables  $t, x, y, z, l, g, sum$. Based on the program assignments, we need $t \preceq x \prec y \preceq l; y \preceq z; t \preceq z \preceq y; x \preceq sum; g \preceq sum$.
Variables y and z form an equivalence class, since $y \preceq z$ and $z \preceq y$. Notice that some variables are not ordered, for example $t$ and $g$. This gives some freedom in choosing the total preorder, any will work. Let us have $t \prec g \prec x \prec sum \prec y \sim z \prec l$. This gives an order on the equivalence classes $\{t\} \prec \{g\} \prec \{x\} \prec \{sum\} \prec \{y, z\} \prec \{l\}$.
The ordering on monomials then considers the classes of non-finite program variables, i.e. all equivalence classes except $\{t\}$. A monomial is then assigned a sequence of degrees with respect to each equivalence class, e.g. $z \to (0,0,0,1,0)$, $tz\to (0,0,0,1,0)$, and $x^2 \to (0,0,2,0,0)$, with monomials ordered reverse-lexicographically with respect to their sequence. For the monomials in this example: $ 1 \sim t \prec x \sim tx \prec x^2 \sim tx^2 \prec y \sim z \sim ty \sim tz$.

After computing all necessary recurrence equations, we are faced with a system of linear recurrences of the expected values of the monomials $z$, $tx$, $ty$, $x$, $y$, $tx^2$, $t$, $tz$, $x^2$ and $1$ with recurrence matrix
\begin{equation*}\label{eq:example-system}
\footnotesize
\left[\begin{matrix}1 & - \frac{1}{6} & - \frac{1}{2} & 0 & 0 & - \frac{1}{6} & \frac{1}{12} & \frac{1}{6} & 0 & 0\\0 & -1 & 0 & 1 & 0 & 0 & 0 & 0 & 0 & 0\\0 & 0 & -1 & 0 & 1 & 0 & 0 & 0 & 0 & 0\\0 & 0 & 0 & 1 & 0 & 0 & \frac{3}{2} & 0 & 0 & 0\\0 & \frac{1}{3} & 0 & 0 & 1 & \frac{1}{3} & - \frac{1}{6} & - \frac{1}{3} & 0 & 0\\0 & 0 & 0 & 0 & 0 & -1 & 0 & 0 & 1 & 0\\0 & 0 & 0 & 0 & 0 & 0 & -1 & 0 & 0 & 1\\1 & 0 & 0 & 0 & 0 & 0 & 0 & -1 & 0 & 0\\0 & 3 & 0 & 0 & 0 & 0 & \frac{5}{2} & 0 & 1 & 0\\0 & 0 & 0 & 0 & 0 & 0 & 0 & 0 & 0 & 1\end{matrix}\right].
\end{equation*}

Using any computer algebra system, such as \texttt{sympy}\footnote{\url{https://www.sympy.org}}, we arrive at the closed-form solution for the expected value of $z$, $\E(z_n)$, parameterized by the loop counter $n$:
\begin{flalign*}
\footnotesize
\begin{split}
\E(z_n) &= \frac{883}{32} + \frac{29n}{16} - \frac{201}{20} 2^{-n} 6^{\frac{n}{2}} - \frac{67}{20} 2^{-n} 6^{\frac{1+n}{2}} - \frac{37}{10} 3^{-n} 6^{\frac{n}{2}} \\ 
&- \frac{37}{20} 3^{-n} 6^{\frac{1+n}{2}} - \frac{201}{20} 6^{\frac{n}{2}} \left(\frac{-1}{2}\right)^n
- \frac{37}{10} 6^{\frac{n}{2}} \left(\frac{-1}{3}\right)^n + \frac{67}{20} 6^{\frac{1+n}{2}} \left(\frac{-1}{2}\right)^n \\
&+ \frac{37}{20} 6^{\frac{1+n}{2}} \left(\frac{-1}{3}\right)^n + \frac{9}{16} n \left(-1\right)^n + \frac{29}{32} \left(-1\right)^n + \frac{9}{16} n^2.
\end{split}
\end{flalign*}

This example highlights the strength of algebraic techniques for probabilistic program analysis in comparison to constraint-based methods employing templates.
We are not aware of any template-based method able to handle functions of the complexity of $\E(z_n)$.
Our tool \Tool{} is able to find the closed-form of $\E(z_n)$ in under one second (see Section~\ref{sec:evaluation}).
\end{example}

\subsection{Guarded Loops}\label{subsec:guarded-loops}

When we model a guarded probabilistic loop \lstinline{while $\phi$: $\dots$} as an infinite loop \lstinline{while $\star$: if $\phi$: $\dots$}, we impose the same restrictions on $\phi$ as on if-conditions (that means $\phi$ only contains finite variables) to guarantee computability and correctness.
The $k$th moment of $x$ after termination is then given by
\begin{equation}\label{eq:after-term}
    \lim_{n\to\infty} \E(x_n^k \mid \lnot\phi_n) = \lim_{n\to\infty} \frac{\E(x_n^k \cdot [\lnot\phi_n])}{\E([\lnot\phi_n])}.
\end{equation}
If the limit exists, we can use standard methods from computer algebra to compute it, as the (higher) moments our approach computes are given as exponential polynomials~\cite{Gruntz1996}.

\pagebreak
\begin{example}
Consider the following loop, in which $x$ after termination is geometrically distributed with parameter $\nicefrac{1}{2}$:

\noindent
\begin{minipage}{\linewidth}
\centering
\begin{minipage}{0.4\linewidth}
    \begin{lstlisting}
    x, stop = 0, 0
    while stop == 0:
        stop=Bernoulli(1/2)
        x=x+1
    end
    \end{lstlisting}
\end{minipage}
\end{minipage}
With Equation~\ref{eq:after-term} and the techniques from this section we get:
\begin{flalign*}
    \footnotesize
    \E(x) = \lim_{n\to\infty} \E(x_n \mid stop_n = 1) 
    = \lim_{n\to\infty} \frac{\E(x_n \cdot stop_n)}{\E(stop_n)} 
    = \lim_{n\to\infty} \frac{-n 2^{-n} - 2^{1 - n} + 2}{1 - 2^{-n}} = 2.
\end{flalign*}
\end{example}

Moreover, whenever the variables in the loop guard are not probabilistic, traditional techniques can be applied to determine the number of loop iterations n which can then be plugged into the (higher) moments computed by our work.
Apart from the guarded loops, many systems show the type of infinite behavior naturally modeled with infinite loops, such as probabilistic protocols or dynamical systems.

\subsection{Infinite If-Conditions}\label{ssec:infinite-if}

Theorem~\ref{thm:moment-computability} on moment-computability requires the variables in all if-conditions to be finite.
Nevertheless, in some cases, if-conditions containing infinite variables can be handled by our approach.
Let $\Program$ be a probabilistic loop containing an if-statement with condition \lstinline{$F$ and $I$} where $F$ contains only finite variables and $I$ contains infinite variables.
Without loss of generality, no variable in $I$ is assigned in or after the if-statement.
Let the transformation removing $I$ be defined by
\noindent
\begin{minipage}{\linewidth}
    \centering
    \begin{minipage}{0.25\linewidth}
    \begin{lstlisting}
    if $F$ and $I$:$Branch$
    \end{lstlisting}
    \end{minipage}
    \begin{minipage}{0.1\linewidth}
    $\transmap$
    \end{minipage}
    \begin{minipage}{0.25\linewidth}
    \begin{lstlisting}
    t = Bernoulli($p$)
    if $F$ and t == 1:
    $Branch$ end
    \end{lstlisting}
    \end{minipage}
\end{minipage}
where \emph{t} is a fresh variable and $p := \P(I)$ (potentially symbolic).
Then, the transformation preserves the distributions of all $x \in \Vars(\Program)$ under the following assumptions:
\begin{enumerate}
    \item\label{enum:inf-cond:iter} $I$ is iteration independent, meaning for every variable $x$ in $I$ neither $x$ nor any variable $x$ depends on (as of Definition~\ref{def:dependency}) has a self-dependency.
    \item\label{enum:inf-cond:cond} $I$ is statistically independent from $F$ and all conditions in $Branch$.
    \item\label{enum:inf-cond:ass} For every assignment $A$ in $Branch$ and every variable $x$ in $A$ which has been assigned before A, it holds that $I$ and $x$ are statistically independent.
\end{enumerate}
Assumption~\ref{enum:inf-cond:iter} ensures that $\P(I) = \E([I])$ is constant.
Assumption~\ref{enum:inf-cond:cond} and \ref{enum:inf-cond:ass} further ensure that $\E([I])$ can always be \lq\lq{}pulled out\rq\rq{} (that means $\E([I]x) = \E([I])\E(x)$) in the construction of the recurrences. Assumptions~\ref{enum:inf-cond:iter}-\ref{enum:inf-cond:ass} can often be checked automatically.

\begin{example}
Consider the statement \lstinline{if g < 1/2: sum=sum+x} of the program from Figure~\ref{fig:running-example}, where the value of $g$ is drawn from a standard normal distribution.
In this case, the transformation's parameter $p$ represents $\P(\Normal(0,1) < \nicefrac{1}{2})$, but is left symbolic for the moment computation. 
The integral $\P(\Normal(0,1) < \nicefrac{1}{2})$ can be solved separately and the result be substituted for $p$.
\end{example}

\subsection{On the Necessity of the Conditions Ensuring Moment-Computability}\label{ssec:necessity}

Theorem~\ref{thm:moment-computability} states two conditions that are sufficient to ensure that the closed-forms of the program variables' higher moments always exist and are computable.
\emph{Condition 1 enforces that there is no variable with potentially infinite values with a polynomial self-dependency.}
\emph{Condition 2 demands that all variables appearing in if-conditions are finite.}
Our approach for computing the moments of variables of probabilistic loops can handle precisely the programs that satisfy these two conditions.
We argue that both conditions are necessary, in the sense that if either of the conditions does not hold, the existence or computability of the variable moments' closed-forms as exponential polynomials cannot be guaranteed for all programs from our program model when one or both conditions are removed from Theorem~\ref{thm:moment-computability}.

\paragraph{Condition 1.}
Relaxing condition 1 of Theorem~\ref{thm:moment-computability} means that we allow for polynomial self-dependencies of non-finite variables.
The \emph{logistic map}~\cite{May1976} is a quadratic first-order recurrence defined by $x_{n+1} = r \cdot x_n (1 - x_n)$ and well-known for its chaotic behavior.
A famous fact about the logistic map is that it does not have an analytical solution for most values of $r$~\cite{Maritz2020}.
By neglecting condition 1, we can easily devise a loop modeling the logistic map:

\noindent
\begin{minipage}{\linewidth}
\centering
\begin{minipage}{0.4\linewidth}
    \begin{lstlisting}
    while $\star$:
        x = r$\cdot$x(1-x)
    end
    \end{lstlisting}
\end{minipage}
\end{minipage}

The value of the program variable $x$ after iteration $n$ is equal to the $n$th term of the logistic map.
This means, for most values of $r$ and initial values of $x$, there does not exist an analytical closed-form solution for the program variable $x$.
Moreover, our counter-example illustrates that condition 1 is necessary already for programs with a single variable and without stochasticity and if-statements.

\paragraph{Condition 2.}
Loosening condition 2 of Theorem~\ref{thm:moment-computability} and allowing for non-finite variables in if-conditions renders our programming model Turing-complete.
Intuitively, one can model a Turing machine's tape with two variables $\texttt{l}$ and $\texttt{r}$ such that the binary representation of $\texttt{l}$ represents the tape's content left of the read-write-head.
The binary representation of $\texttt{r}$ represents the tape's content at the position of the read-write-head and towards the right.
The least significant bit of $\texttt{r}$ is the current symbol the Turing machine is reading.
We can extract the least significant bit of $\texttt{r}$ in our programming model (and neglecting condition 2) by introducing a variable $\texttt{lsb}$ and using a single if-statement involving non-finite variables: whenever the loop changes the value of $\texttt{r}$, we set $\texttt{lsb} := \texttt{r}$. The while-loop's body is of the form \lq\lq\lstinline{if lsb > 1: lsb=lsb-2 else $\mathit{transitions}$ end}\rq\rq.
The Turing machine's transition table can be encoded using if-statements.
Writing and shifting can be accommodated for by multiplying by $2$ or $\nicefrac{1}{2}$ and using addition and subtraction.
The Turing machine's state can be modelled by a single finite variable.
Therefore, by dropping condition 2, being able to model the program variables by linear recurrences would give rise to a decision procedure for the Halting problem:
assume we introduce a variable $\texttt{terminated}$ which is initialized to $0$ before the loop and set to $1$ whenever the Turing-machine terminates.
If $\texttt{terminated}$ can be modelled by a linear recurrence of order $k$, it suffices to check the first $k$ values of the recurrence to determine whether or not $\texttt{terminated}$ is always $0$ \cite{Kauers2011} and the Turing-machine does not terminate.
As the Halting problem is well-known to be undecidable, condition 2 is necessary to guarantee that the program variables can be modelled by linear recurrences, even without stochasticity and polynomial arithmetic.

\begin{remark}[Sequential \& Nested Loops]
Our program model consists of single non-nested loops.
Sequential loops can be analyzed one by one with the same techniques as presented in this section.
For nested loops, one could design a program transformation transforming a nested loop into a non-nested loop and then apply the techniques presented in this section.
Alternatively, we conjecture that the approach presented for guarded loops (Section~\ref{subsec:guarded-loops}) could be used to first compute the moments of the most inner loops and then use the obtained information to compute the moments of the outer loops.
The main challenge lies in ensuring the moment-computability conditions for the outer loops (Theorem~\ref{thm:moment-computability}) once the inner loops have been analyzed.
\end{remark}

%% file: 06-use-cases.tex
\section{Use-Cases of Higher Moments}\label{sec:moments}

For probabilistic loops, computing closed-forms of the variables' (higher) moments poses a technique for synthesizing quantitative invariants:
Given a program variable $x$ and a closed-form $f(n)$ of its $k$th moment, the equation $\E(x_n^k) - f(n) = 0$ is an invariant.
Moreover, closed-forms of raw moments can be converted into closed-forms of \emph{central} moments, such as variance, skewness or kurtosis (cf. Section~\ref{ssec:prob-theory}).
In addition, this section provides hints on two further use-cases of higher moments of probabilistic loops: (i) deriving tail probabilities (Section~\ref{ssec:tailProbs}) and (ii) inferring distributions of random variables from their moments (Section~\ref{ssec:mom-dist}).

\subsection{From Moments to Tail Probabilities}\label{ssec:tailProbs}

Tail probabilities measure the probability that a random variable surpasses some value.
The mathematical literature contains several inequalities providing upper- or lower bounds on tail probabilities given (higher) moments~\cite{Boucheron2013}.
Two examples are \emph{Markov's inequality} for upper and the \emph{Paley-Zygmund inequality} for lower bounds.

\begin{theorem}[Markov's Inequality]\label{thm:markov-inequality}
Let $X$ be a non-negative random variable, and~$t \geq 0$, then
$$\P(X \geq t) \leq \frac{\E(X^k)}{t^k}.$$
\end{theorem}
\begin{theorem}[Paley-Zygmund Inequality]\label{thm:paley-zygmund}
Let $X$ be a random variable with $X \geq t$ almost-surely. Then
$$\P(X > t) \geq \frac{(\E(X) - t)^2}{\E(X^2) - 2t\E(X) + t^2}.$$
\end{theorem}

\begin{example}\label{ex:tail-probabilities}
For \emph{Herman's Self-Stabilization} program from Figure~\ref{fig:herman3} almost-surely $tokens \in \{ 0, 1, 2, 3 \}$.
With the techniques from previous sections, we can compute the first two moments $\E(tokens_n) = 1 + 2 \cdot 4^{-n}$ and $\E(tokens^2_n) = 1 + 8 \cdot 4^{-n}$.
Markov's inequality (Theorem~\ref{thm:markov-inequality}) gives us the upper bound $\P(tokens_n \geq 2) \leq \nicefrac{1}{2} + 4^{-n}$ using the first moment and $\P(tokens_n \geq 2) \leq \nicefrac{1}{4} + 2 \cdot 4^{-n}$ utilizing the second moment.

For the Paley-Zygmund inequality (Theorem~\ref{thm:paley-zygmund}) both the first and the second moment are required, yielding the lower bound $\P(tokens_n \geq 2) = \P(tokens_n > 1) \geq 4^{-n}$.
The theorem's precondition that almost-surely $tokens \geq 1$ might not be apparent at first sight.
We take this for granted for now and will clarify this fact in Example~\ref{ex:herman:distribution}.
\end{example}

Markov's inequality and the Paley-Zygmund inequality are just two examples showing that our technique for moment computation can be leveraged for further program analysis using known results from probability theory.
Our approach computes the \emph{exact moments} of variables in probabilistic loops instead of just approximations or bounds on moments.
This enables our technique to be readily combined with results from probability theory that require exact moments.

\subsection{From Moments to Distributions}\label{ssec:mom-dist}

For finite random variables, their full distribution can be recovered from finitely many moments.
More precisely, given a random variable $X$ with $m$ possible values, the distribution of $X$ can be recovered from its first $m{-}1$ moments, as the following theorem states: 

\begin{theorem}\label{thm:DistributionFromMoments}
Let $X$ be a random variable over $\{a_1, \dots, a_m\}$ and $p_i := P(X = a_i)$.
The values $p_i$ are the solutions of the system of linear equations given by $\sum_{i=1}^m p_i a_i^j = \E(X^j)$ for $0 \leq j < m$.
\end{theorem}
\begin{proof}
There are $m$ unknowns $p_i$ for $1 \leq p_i \leq m$.
Note that all $a_i$ are constant and that the first $m{-}1$ raw moments of $X$ are fixed.
Using the definition of raw moments, we get $m$ linear equations $\sum_{i=1}^m p_i a_i^j = \E(X^j)$ for $0 \leq j < m$.
The solutions of the system of $m$ linear equations are the values $p_i$ that determine the distribution of $X$.
\end{proof}

\begin{example}\label{ex:herman:distribution}
Consider \emph{Herman's Self-Stabilization} program from Figure~\ref{fig:herman3}.
In Example~\ref{ex:tail-probabilities} we obtained upper and lower bounds for tail probabilities of the $tokens$ variable using the first one or two moments.
With the first \emph{three} moments we can fully recover the distribution of the $tokens$ variable.
We have that $tokens \in \{0,1,2,3\}$.
Let $p_i := \P(tokens_n = i)$ for $0 \leq i \leq 3$.
By Theorem~\ref{thm:DistributionFromMoments}, we get the following system:
\begin{flalign*}
\begin{split}
p_0 + p_1 + p_2 + p_3 &= \E(tokens_n^0) = 1, \\
p_1 + 2p_2 + 3p_3 &= \E(tokens_n) = 1 + 2 \cdot 4^{-n}, \\
p_1 + 4p_2 + 9p_3 &= \E(tokens_n^2) = 1 + 8 \cdot 4^{-n}, \\
p_1 + 8p_2 + 27p_3 &= \E(tokens_n^3) = 1 + 26 \cdot 4^{-n}. 
\end{split}
\end{flalign*}
The solution can be obtained using standard techniques and tools, yielding $p_0 = 0; p_1 = 1 - 4^{-n}; p_2 = 0; p_3 = 4^{-n}$.

Note that probabilities are given as functions of the loop iteration $n$.
Moreover, the solution shows that almost-surely $tokens \geq 1$, which we assumed to be true in Example~\ref{ex:tail-probabilities}.
\end{example}

\begin{figure}
    \centering
    \begin{subfigure}{0.48\linewidth}
        \centering
        \includegraphics[trim={0.5cm, 0, 0.5cm, 0.5cm},clip,width=0.9\linewidth]{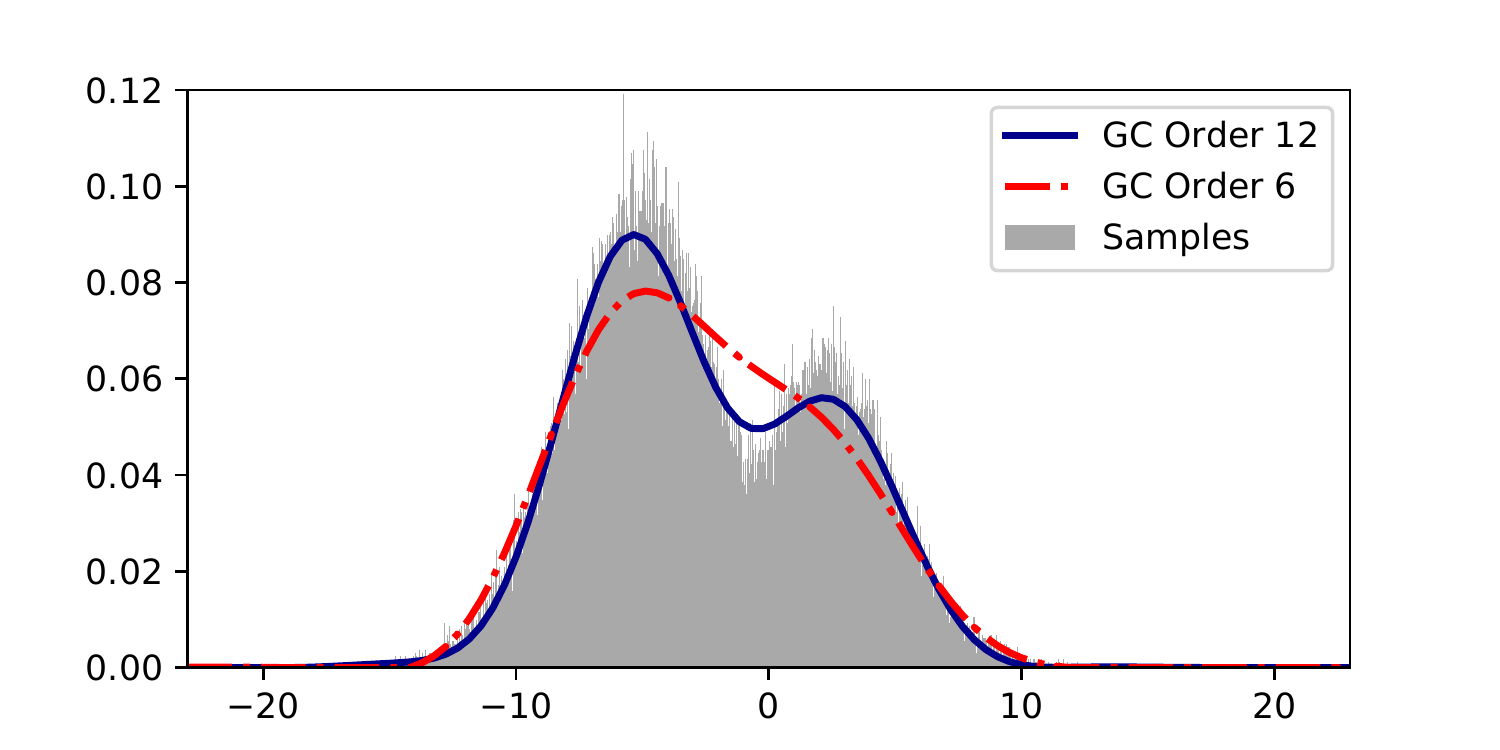}
    \end{subfigure}
    \begin{subfigure}{0.48\linewidth}
        \centering
        \includegraphics[trim={0.5cm, 0, 0.5cm, 0.5cm},clip,width=0.9\linewidth]{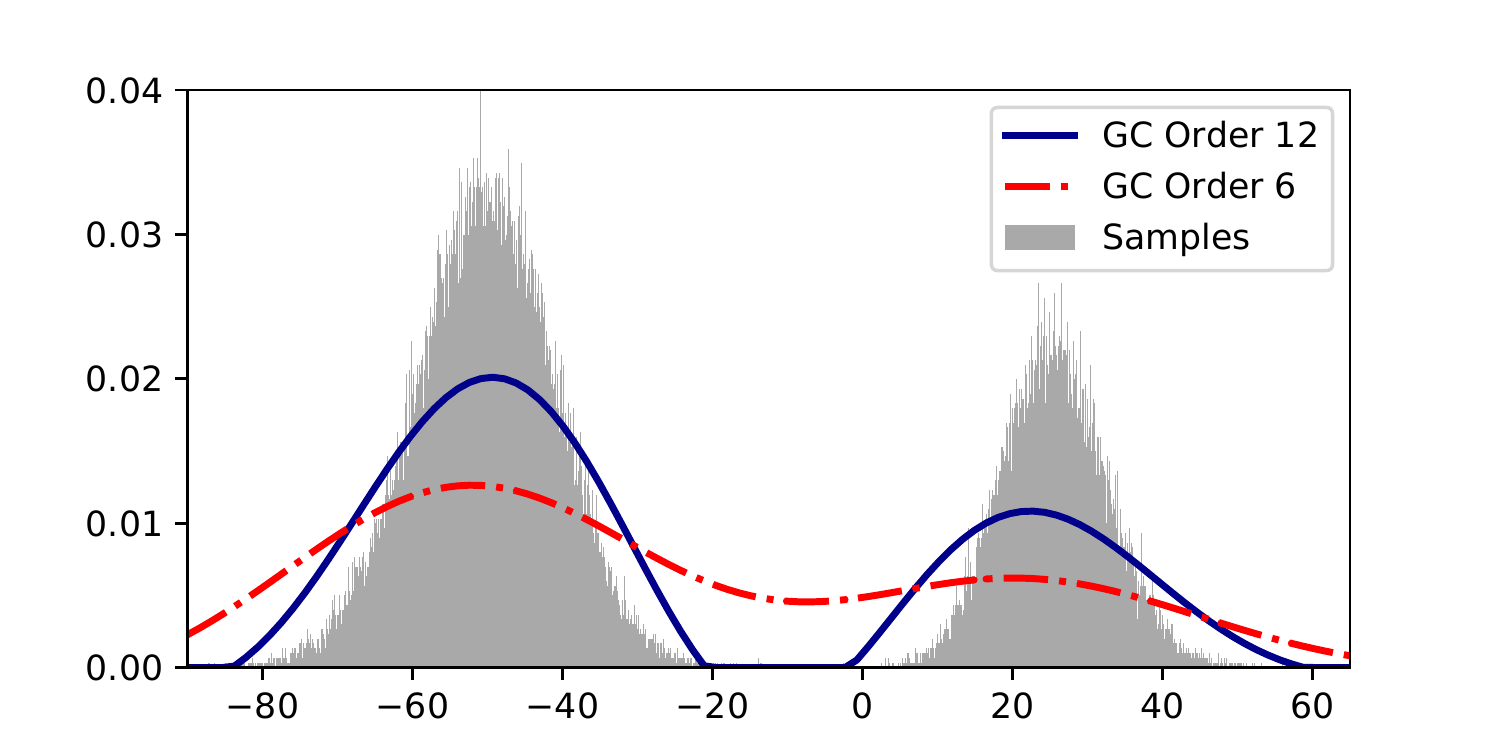}
    \end{subfigure}
    \caption{
    The empirical density of program variable $x$ for the benchmark \emph{Bimodal} (cf. Table~\ref{tbl:higher-moments}) and loop iterations $10$ (left) and~$100$ (right) obtained by~$10^5$ samples, together with two approximations using the Gram-Charlier A Series with $6$ (red dashed lines) and $12$ (blue solid lines) moments.
    }
    \label{fig:gram-charlier}
\end{figure}

The distributions of program variables with potentially infinitely many values, including continuous variables, cannot be, in general, fully reconstructed from finitely many moments.
However, expansions such as the \emph{Gram–Charlier A Series}~\cite{Kolassa2013} can be used to approximate a probability density function using finitely many moments.
Figure~\ref{fig:gram-charlier} illustrates how exact moments computed by our approach can be used to approximate unknown probability density functions of program variables.
While Figure~\ref{fig:gram-charlier} shows approximations for specific loop iterations, we emphasize that the symbolic nature of our approach allows for approximating the densities of program variables for all -- potentially infinitely many -- loop iterations simultaneously.
Therefore, our technique is constant in the number of loop iterations, whereas sampling has linear complexity.
We compute the approximations from Figure~\ref{fig:gram-charlier} for all infinitely many loop iterations in $\sim$$22$ seconds, with the experimental setup from Section~\ref{sec:evaluation}.
In comparison, sampling the loop $10^5$ times takes $\sim$$3.6$ minutes for loop iteration $10$ and $\sim$$33$ minutes for loop iteration $100$.

%% file: 07-experiments.tex
\pagebreak
\section{Implementation and Evaluation}
\label{sec:evaluation}

\paragraph{Implementation.} The program transformations (Section \ref{sec:model}) and (higher) moment computation (Section~\ref{sec:computing-moments}) are implemented in the new tool \Tool{}~\cite{PolarArtifact}.
The experiments can be reproduced using the corresponding artifact\footnote{\url{https://doi.org/10.5281/zenodo.7055030}}.
For automatically inferring finiteness of program variables, we use a standard approach based on \emph{abstract interpretation}.
\Tool{} is implemented in \texttt{python3}, consisting of $\sim$$3300$ LoC, and uses the packages \texttt{sympy}\footnote{\url{https://www.sympy.org}} and \texttt{symengine}\footnote{\url{https://github.com/symengine}} for symbolic manipulation of mathematical expressions.
Together with all our benchmarks, \Tool{} is publicly available at \url{https://github.com/probing-lab/polar}.

\paragraph{Experimental Setting and Evaluation.} The evaluation of our work is split into three parts.
First, we evaluate \Tool{} on the ability of computing higher moments for $15$ probabilistic programs exhibiting different characteristics (Section~\ref{sec:exp:higherMoment}). 
Second, we compare \Tool{} to the exact tool \Mora{}~\cite{Bartocci2020a} which computes so-called \emph{moment-based invariants} for a subset of our programming model (Section~\ref{sec:comp:exactmethods}).
Third, we compare our tool to approximate methods estimating program variable moments by confidence intervals through sampling (Section~\ref{sec:sampling}).
All experiments have been run on a machine with a \SI{2.6}{GHz} Intel i7 (Gen 10) processor and \SI{32}{GB} of RAM.
Runtime measurements are averaged over $10$ executions.

\begin{table}[h]
    \footnotesize
    \caption{Evaluation of \Tool{} on $15$ benchmarks. All times are in seconds. \#V = number of variables in benchmark; C = benchmark contains circular dependencies; If = benchmark contains if-statements; S = benchmark contains symbolic constants; INF = benchmark's states space is infinite; CONT = benchmark's state space is continuous; Moment = Moment to compute; RT = Total runtime.}
	\label{tbl:higher-moments}
    
	\begin{tabular}{lcccccc@{\hskip 3em}cccc}
	    \toprule
		Benchmark & \#V & C/If/S/INF/CONT & Moment & RT \\
		\midrule
		Running-Example (Fig.~\ref{fig:running-example}) & 7 & \cmark / \cmark / \cmark / \cmark / \cmark & $\E(z)$ & 0.67 \\
		\midrule
		Herman-3 & 10 & \cmark / \cmark / \xmark / \xmark / \xmark & $\E(\text{tokens}^3)$ & 0.58 \\
		\midrule
		Las-Vegas-Search & 3 & \xmark / \cmark / \xmark / \cmark / \xmark & $\E(\text{found}^{20})$ & 0.36 \\
		\midrule
		Pi-Approximation & 4 & \xmark / \cmark / \xmark / \cmark / \cmark & $\E(\text{count}^3)$ & 0.47 \\
		\midrule
		50-Coin-Flips & 101 & \xmark / \cmark / \xmark / \xmark / \xmark & $\E(\text{total})$ & 0.91 \\
		\midrule
		Gambler-Ruin-Momentum & 4 & \cmark / \xmark / \cmark / \cmark / \xmark & $\E(x^3)$ & 2.89 \\
		\midrule
		Hawk-Dove-Symbolic & 5 & \xmark / \cmark / \cmark / \cmark / \xmark & $\E(\text{p1bal}^4)$ & 2.00 \\
		\midrule
		Variable-Swap & 4 & \cmark / \xmark / \xmark / \cmark / \cmark & $\E(x^{30})$ & 2.42 \\
		\midrule
		Retransmission-Protocol & 4 & \xmark / \cmark / \cmark / \cmark / \xmark & $\E(\text{fail}^3)$ &  1.62 \\
		\midrule
		Randomized-Response & 7 & \xmark / \cmark / \cmark / \cmark / \xmark & $\E(\text{p1}^3)$ &  0.59 \\
		\midrule
		Duelling-Cowboys & 4 & \cmark / \cmark / \cmark / \xmark / \xmark & $\E(\text{ahit})$ & 1.14 \\
		\midrule
		Martingale-Bet & 4 & \xmark / \cmark / \cmark / \cmark / \xmark & $\E(\text{capital}^3)$ & 8.44 \\
		\midrule
		Bimodal & 5 & \xmark / \cmark / \xmark / \cmark / \cmark & $\E(\text{x}^{10})$ & 4.50 \\
		\midrule
		DBN-Umbrella & 2 & \xmark / \cmark / \cmark / \xmark / \xmark & $\E(\text{umbrella}^5)$ & 0.77 \\
		\midrule
		DBN-Component-Health & 3 & \xmark / \cmark / \xmark / \xmark / \xmark & $\E(\text{obs}^5)$ & 0.26 \\
		\bottomrule
	\end{tabular}
\end{table}

\subsection{Experimental Results with Higher Moments}\label{sec:exp:higherMoment}

Table~\ref{tbl:higher-moments} shows the evaluation of \Tool{} on the program from Figure~\ref{fig:running-example} and $14$ benchmarks which are either
from the literature on probabilistic programming \cite{Kwiatkowska2012} (\emph{Herman-3}), \cite{McIver2005} (\emph{Duelling-Cowboys}), \cite{Gretz2013} (\emph{Martingale-Bet}), \cite{Chakarov2014} (\emph{Hawk-Dove-Symbolic}, \emph{Variable-Swap}), \cite{Barthe2016} (\emph{Gambler-Ruin-Momentum}), \cite{Batz2021} (\emph{Retransmission-Protocol}), differential privacy schemes~\cite{Warner1965} (\emph{Randomized-Response}),
Dynamic Bayesian Networks (\emph{DBN-Umbrella}, \emph{DBN-Component-Health}) or
well-known stochastic processes (\emph{Las-Vegas-Search}, \emph{Pi-Approximation}, \emph{Bimodal}).
The benchmarks \emph{Retransmission-Protocol} and \emph{Hawk-Dove-Symbolic} were further generalized from their original definition by replacing concrete numbers with symbolic constants.
This makes these benchmarks only harder as solutions to the generalized versions are solutions for the concretizations.
Table~\ref{tbl:higher-moments} illustrates that \Tool{} can compute higher moments for various probabilistic programs exhibiting different features, like circular variable dependencies, if-statements, and symbolic constants with finite, infinite, continuous, and discrete state spaces.
Moreover, the table shows that the number of program variables is \emph{not} the primary factor for the complexity of computing moments.
For instance, the benchmarks \emph{50-Coin-Flips} and \emph{Duelling-Cowboys} have $101$ and $4$ program variables respectively.
Nevertheless, the runtimes for computing first moments for the two benchmarks only differ by \SI{0.23}{s}.
The complexity of computing moments lies in the complexity of the resulting systems of recurrences which depend on the concrete features present in the benchmarks like specific variable dependencies, symbolic constants, or degrees of polynomials.

\subsection{Experimental Comparison to Exact Methods}\label{sec:comp:exactmethods}

\begin{table}
    \footnotesize
    \caption{Comparison of \Tool{} to \Mora{}. The runtimes are in seconds per tool, benchmark, and moment. For \Tool{} the comparison contains in brackets the seconds spent on parsing, normalizing, and type inference.}
	\label{tbl:mora}
    
    \begin{minipage}[t]{0.5\textwidth}
    \vspace{0pt}
	\begin{tabular}{lcc}
		\toprule
		Benchmark & \Mora & \Tool \\
		\midrule
		
		COUPON
		\hfill
		\begin{tabular}{l}
            $\E(c)$ \\ $\E(c^2)$ \\ $\E(c^3)$
		\end{tabular}
		&
		\setlength\tabcolsep{0pt}
		\begin{tabular}{l}
            $0.25$ \\ $0.27$ \\ $0.29$
		\end{tabular}
		&
		\setlength\tabcolsep{0pt}
		\begin{tabular}{l}
            $0.29$ \\ $0.29$ \\ $0.29$
		\end{tabular}
		$\color{darkgray}(0.07)$
		\\
		\midrule
		
		COUPON4
		\hfill
		\begin{tabular}{l}
            $\E(c)$ \\ $\E(c^2)$ \\ $\E(c^3)$
		\end{tabular}
		&
		\setlength\tabcolsep{0pt}
		\begin{tabular}{l}
            $0.71$ \\ $0.87$ \\ $1.22$
		\end{tabular}
		& 
		\setlength\tabcolsep{0pt}
		\begin{tabular}{l}
            $0.36$ \\ $0.36$ \\ $0.36$
		\end{tabular}
		$\color{darkgray}(0.08)$
		\\
		\midrule
		
		RANDOM\_WALK\_1D
		\hfill
		\begin{tabular}{l}
            $\E(x)$ \\ $\E(x^2)$ \\ $\E(x^3)$
		\end{tabular}
		&
		\setlength\tabcolsep{0pt}
		\begin{tabular}{l}
            $0.07$ \\ $0.11$ \\ $0.10$
		\end{tabular}
		& 
		\setlength\tabcolsep{0pt}
		\begin{tabular}{l}
            $0.12$ \\ $0.23$ \\ $0.24$
		\end{tabular}
		$\color{darkgray}(0.07)$
		\\
		\midrule
		
		SUM\_RND\_SERIES
		\hfill
		\begin{tabular}{l}
            $\E(x)$ \\ $\E(x^2)$ \\ $\E(x^3)$
		\end{tabular}
		&
		\setlength\tabcolsep{0pt}
		\begin{tabular}{l}
            $0.27$ \\ $0.97$ \\ $2.48$
		\end{tabular}
		&
		\setlength\tabcolsep{0pt}
		\begin{tabular}{l}
            $0.27$ \\ $0.43$ \\ $0.79$
		\end{tabular}
		$\color{darkgray}(0.07)$
		\\
		\midrule
		
		PRODUCT\_DEP\_VAR
		\hfill
		\begin{tabular}{l}
            $\E(p)$ \\ $\E(p^2)$ \\ $\E(p^3)$
		\end{tabular}
		&
		\setlength\tabcolsep{0pt}
		\begin{tabular}{l}
            $0.37$ \\ $1.41$ \\ $4.03$
		\end{tabular}
		&
		\setlength\tabcolsep{0pt}
		\begin{tabular}{l}
            $0.28$ \\ $0.46$ \\ $0.97$
		\end{tabular}
		$\color{darkgray}(0.08)$
		\\
		\midrule
		
		RANDOM\_WALK\_2D
		\hfill
		\begin{tabular}{l}
            $\E(x)$ \\ $\E(x^2)$ \\ $\E(x^3)$
		\end{tabular}
		&
		\setlength\tabcolsep{0pt}
		\begin{tabular}{l}
            $0.10$ \\ $0.21$ \\ $0.17$
		\end{tabular}
		&
		\setlength\tabcolsep{0pt}
		\begin{tabular}{l}
            $0.12$ \\ $0.24$ \\ $0.24$
		\end{tabular}
		$\color{darkgray}(0.08)$
		\\
		\midrule
		
		BINOMIAL(p)
		\hfill
		\begin{tabular}{l}
            $\E(x)$ \\ $\E(x^2)$ \\ $\E(x^3)$
		\end{tabular}
		&
		\setlength\tabcolsep{0pt}
		\begin{tabular}{l}
            $0.12$ \\ $0.32$ \\ $0.79$
		\end{tabular}
		&
		\setlength\tabcolsep{0pt}
		\begin{tabular}{l}
            $0.25$ \\ $0.29$ \\ $0.44$
		\end{tabular}
		$\color{darkgray}(0.07)$
		\\
		\bottomrule
	\end{tabular}
	
	\end{minipage}\hfill
	\begin{minipage}[t]{0.5\textwidth}
	\vspace{0pt}
	\begin{tabular}{lcc}
		\toprule
		Benchmark & \Mora & \Tool \\
		\midrule
		
		STUTTERING\_A
		\hfill
		\begin{tabular}{l}
            $\E(s)$ \\ $\E(s^2)$ \\ $\E(s^3)$
		\end{tabular}
		&
		\setlength\tabcolsep{0pt}
		\begin{tabular}{l}
            $0.29$ \\ $1.13$ \\ $3.32$
		\end{tabular}
		&
		\setlength\tabcolsep{0pt}
		\begin{tabular}{l}
            $0.26$ \\ $0.43$ \\ $0.97$
		\end{tabular}
		$\color{darkgray}(0.07)$
		\\
		\midrule
		
		STUTTERING\_B
		\hfill
		\begin{tabular}{l}
            $\E(s)$ \\ $\E(s^2)$ \\ $\E(s^3)$
		\end{tabular}
		&
		\setlength\tabcolsep{0pt}
		\begin{tabular}{l}
            $0.26$ \\ $0.94$ \\ $2.26$
		\end{tabular}
		&
		\setlength\tabcolsep{0pt}
		\begin{tabular}{l}
            $0.27$ \\ $0.39$ \\ $0.79$
		\end{tabular}
		$\color{darkgray}(0.07)$
		\\
		\midrule
		
		STUTTERING\_C
		\hfill
		\begin{tabular}{l}
            $\E(s)$ \\ $\E(s^2)$ \\ $\E(s^3)$
		\end{tabular}
		&
		\setlength\tabcolsep{0pt}
		\begin{tabular}{l}
            $0.76$ \\ $12.43$ \\ $74.83$
		\end{tabular}
		&
		\setlength\tabcolsep{0pt}
		\begin{tabular}{l}
            $0.40$ \\ $1.94$ \\ $8.19$
		\end{tabular}
		$\color{darkgray}(0.08)$
		\\
		\midrule
		
		STUTTERING\_D
		\hfill
		\begin{tabular}{l}
            $\E(s)$ \\ $\E(s^2)$ \\ $\E(s^3)$
		\end{tabular}
		&
		\setlength\tabcolsep{0pt}
		\begin{tabular}{l}
            $0.76$ \\ $8.19$ \\ $25.67$
		\end{tabular}
		&
		\setlength\tabcolsep{0pt}
		\begin{tabular}{l}
            $0.43$ \\ $1.34$ \\ $4.33$
		\end{tabular}
		$\color{darkgray}(0.07)$
		\\
		\midrule
		
		STUTTERING\_P
		\hfill
		\begin{tabular}{l}
            $\E(s)$ \\ $\E(s^2)$ \\ $\E(s^3)$
		\end{tabular}
		&
		\setlength\tabcolsep{0pt}
		\begin{tabular}{l}
            $0.26$ \\ $1.17$ \\ $3.53$
		\end{tabular}
		&
		\setlength\tabcolsep{0pt}
		\begin{tabular}{l}
            $0.28$ \\ $0.49$ \\ $1.17$
		\end{tabular}
		$\color{darkgray}(0.07)$
		\\
		\midrule
		
		SQUARE
		\hfill
		\begin{tabular}{l}
            $\E(y)$ \\ $\E(y^2)$ \\ $\E(y^3)$
		\end{tabular}
		&
		\setlength\tabcolsep{0pt}
		\begin{tabular}{l}
            $0.30$ \\ $0.88$ \\ $1.98$
		\end{tabular}
		&
		\setlength\tabcolsep{0pt}
		\begin{tabular}{l}
            $0.29$ \\ $0.45$ \\ $0.65$
		\end{tabular}
		$\color{darkgray}(0.07)$
		\\
		\bottomrule
	\end{tabular}
	
	\end{minipage}
\end{table}

To the best of our knowledge, \Mora{} is the only other tool capable of computing higher moments for variables of probabilistic loops without templates -- as described in~\cite{Bartocci2019}.
\Mora{} operates on so-called \emph{Prob-solvable loops} which form a strict subset of our program model (Section~\ref{sec:model}).
Prob-solvable loops do not admit circular variable dependencies, if-statements, or state-dependent distribution parameters.
We compare \Tool{} against \Mora{} on the \Mora{} benchmarks taken from \cite{Kura2019,Chen2015,Chakarov2014,Katoen2010}.
Details can be found in Table~\ref{tbl:mora}. 
The experiments illustrate that \Tool{} can handle all programs and moments that \Mora{} can.
\Mora{}, however, cannot compute any moment for any program in Table~\ref{tbl:higher-moments}.
On simple benchmarks, \Tool{} is slightly slower than \Mora{} due to the constant overhead of the program transformations and type inference to identify finite valued variables.
On complex benchmarks \Tool{} provides a significant speedup compared to \Mora{}.
For instance, for the \emph{STUTTERING_C} benchmark \Tool{} computes the moment $\E(s^3)$ in about $8$ seconds, whereas \Mora{} needs over one minute.

\begin{table}
    \bigskip
    \bigskip
    \footnotesize
    \caption{Comparison of \Tool{} to approximation through sampling. \Tool{} = the tools runtime to compute the precise moment; CI N = an approximated 0.95-CI-interval from N samples; $T_{100.000}$ = the runtime for CI 100.000. The symbolic constant $p$ in \emph{Retransmission-Protocol} is set to $0.9$ and for \emph{Hawk-Dove-Symbolic} we set $v=4$ and $c=8$.}
	\label{tbl:simulation}
	\begin{tabular}{lccccccc}
	    \toprule
		Benchmark &
		Moment &
		\begin{tabular}{l}
            CI 100 \\
    		CI 1.000 \\
    		CI 100.000
		\end{tabular} &
		$T_{100.000}$ &
		$\Tool$ \\
		\midrule
		Running-Example (Fig.~\ref{fig:running-example}) &
		$\E(z_{10})$ &
		\begin{tabular}{l}
            $({-}55.6, 6.55)$ \\
            $({-}58.5, {-}39.3)$ \\
            $({-}45.6, {-}43.7)$
		\end{tabular} &
		545.6s &
		0.67s
		\\
		\midrule
		Retransmission-Protocol &
		$\E(\text{fail}_{10})$ &
		\begin{tabular}{l}
            $(0.09, 0.25)$ \\
    		$(0.11, 0.16)$ \\
    		$(0.109, 0.114)$
		\end{tabular} &
		146.2s &
		0.30s
		\\
		\midrule
		Variable-Swap &
		$\E(y_{10})$ &
		\begin{tabular}{l}
            $(4.47, 5.75)$ \\
    		$(5.06, 5.45)$ \\
    		$(5.50, 5.54)$ \\
		\end{tabular} &
		245.8s &
		0.13s
		\\
		\midrule
		Hawk-Dove-Symbolic &
		$\E(p1bal_{10})$ &
		\begin{tabular}{l}
            $(8.28, 12.3)$ \\
    		$(9.07, 10.5)$ \\
    		$(9.86, 10.0)$
		\end{tabular} &
		347.4s &
		0.27s
		\\
		\bottomrule
	\end{tabular}
\end{table}

\subsection{Experimental Comparison with Sampling}\label{sec:sampling}

For a probabilistic loop with program variable $x$ the moment $\E(x^k_n)$ can be approximated for fixed $k$ and $n$ by sampling $x^k_n$ and calculating the sample average or confidence intervals.
Table~\ref{tbl:simulation} compares \Tool{} to computing confidence intervals by sampling for $k=1$ and $n=10$.
The table shows that our tool is able to compute precise moments in a fraction of the time needed to sample programs to achieve satisfactory confidence intervals.
An advantage of sampling is that it is applicable for any probabilistic loop.
However, by its nature, sampling fails to give any formal guarantees or hard bounds on the approximated moments.
This is critical if the loop body contains branches that are executed with low probability.
If applicable, \Tool{} can provide \emph{exact} moments for symbolic $n$ (and involving other symbolic constants) faster than sampling can establish acceptable approximations.
Moreover, even if the sampling of the loops is sped up by using a more efficient implementation, \Tool{} enjoys complexity theoretical advantages.
The complexity of sampling is linear in both the number of samples and the number of loop iterations.
In contrast, \Tool{} does not need to take multiple samples for higher precision as it symbolically computes the exact moments.
Additionally, our method is \emph{constant} in the number of loop iterations.
With \Tool{}, computing the moment for a specific loop iteration, say $10^5$, just amounts to evaluate the closed-form at $10^5$.

Figure~\ref{fig:simulation} illustrates the importance of higher moments for probabilistic loops.
The first moment provides a center of mass but contains no information on how the mass is distributed around this center.
For this purpose higher moments are essential.

\begin{figure}
    \centering
    \begin{subfigure}[t]{0.32\textwidth}
        \centering
        \includegraphics[trim={0.5cm, 0, 1.6cm, 1cm}, clip, width=\linewidth]{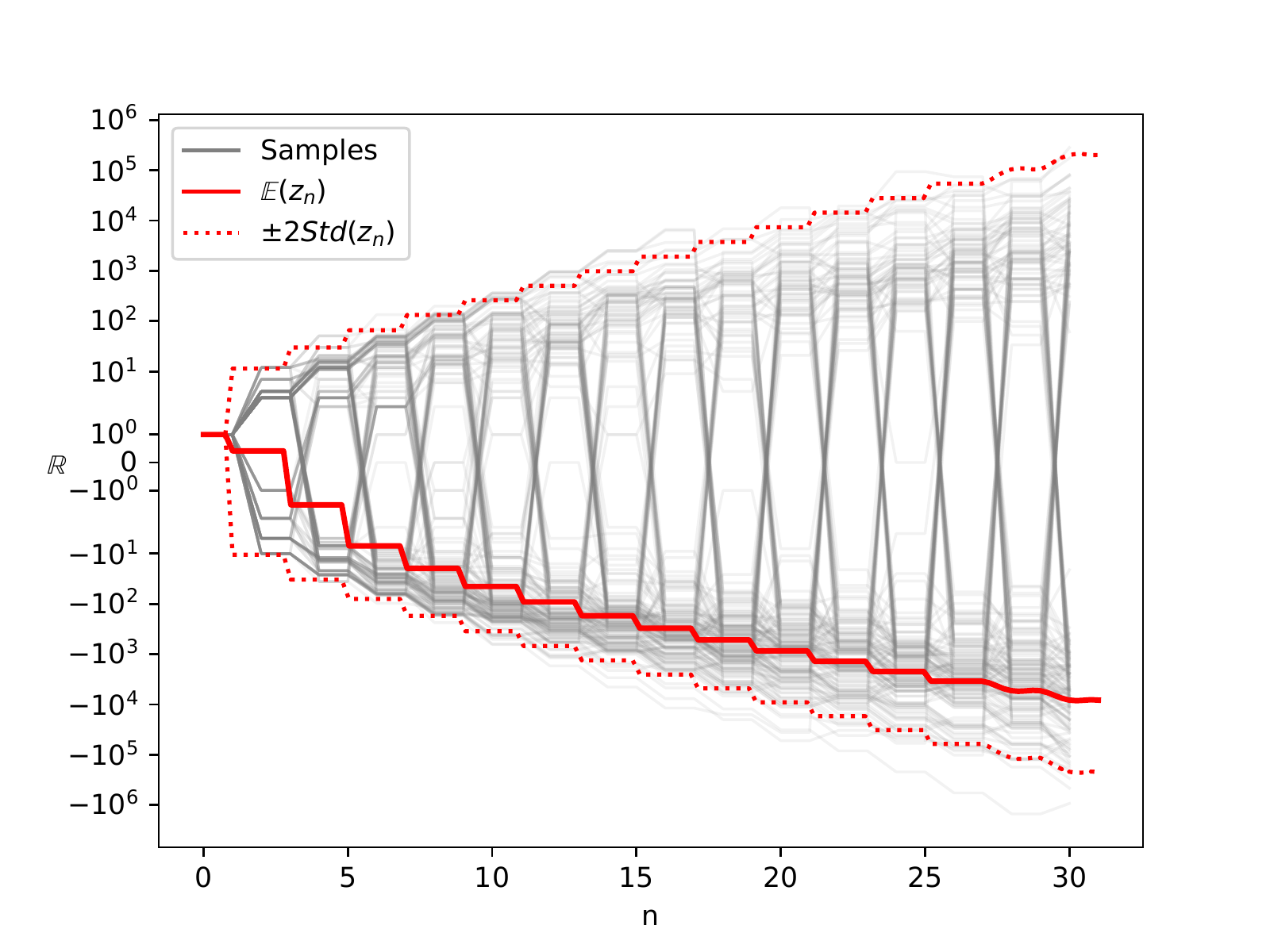}
        \caption{Running-Example (Fig.~\ref{fig:running-example})}
        \label{subfig:running}
    \end{subfigure}
    \begin{subfigure}[t]{0.32\textwidth}
        \centering
        \includegraphics[trim={0.5cm, 0, 1.6cm, 1cm}, clip, width=\linewidth]{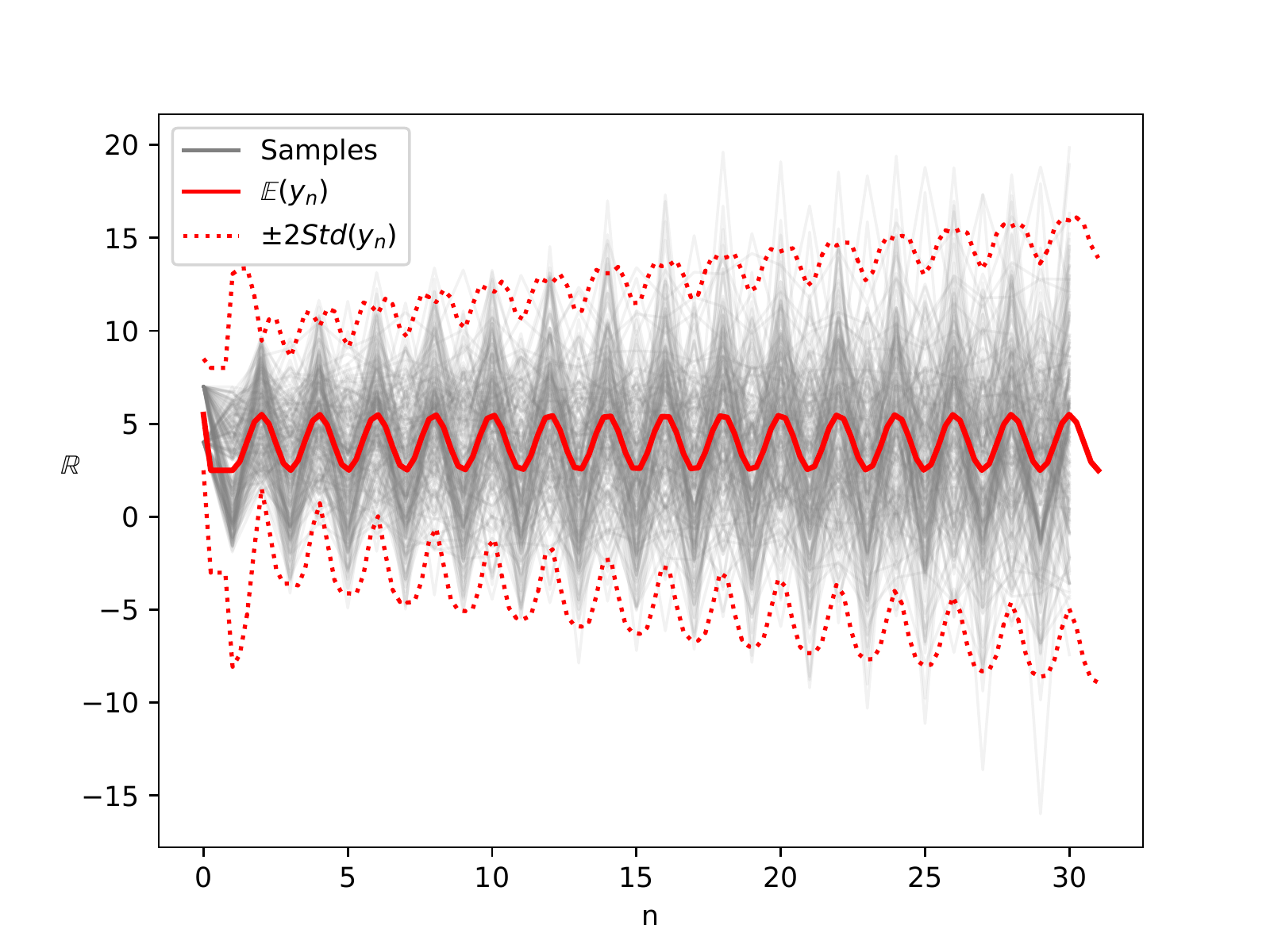}
        \caption{Variable-Swap}
    \end{subfigure}
    \begin{subfigure}[t]{0.32\textwidth}
        \centering
        \includegraphics[trim={0.5cm, 0, 1.6cm, 1cm}, clip, width=\linewidth]{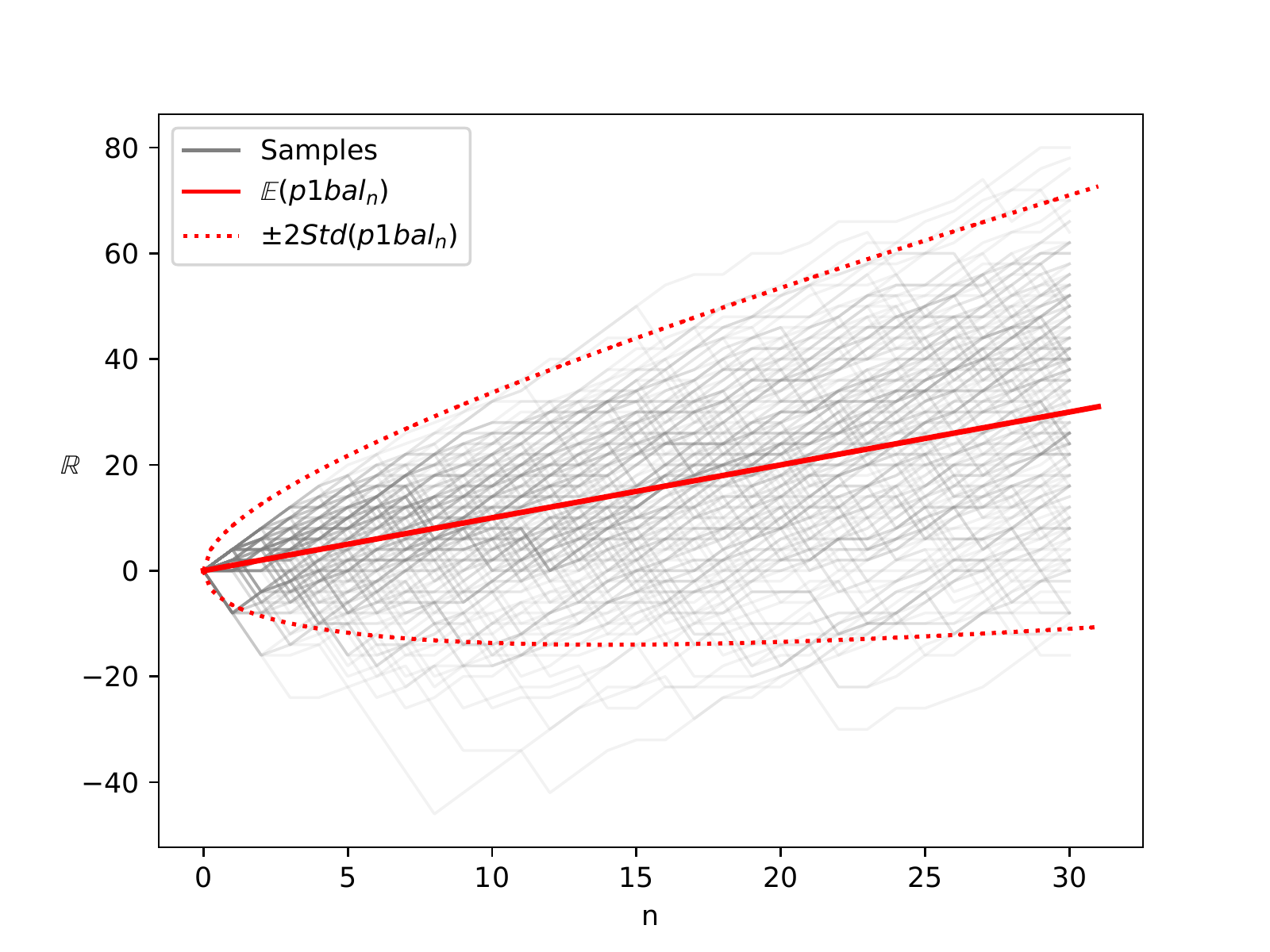}
        \caption{Hawk-Dove-Symbolic}
    \end{subfigure}
    \caption{Samples obtained by simulation plotted together with precise moments computed by \Tool{}. In each benchmark, the thin gray lines are $200$ samples over $30$ iterations. The thick red line is the precise expected value. The dotted red lines are the expected values $\pm$ twice the standard deviation given by the precise first two moments. Figure~\ref{subfig:running} is symmetric log scale.}
    \label{fig:simulation}
\end{figure}

\subsection{Evaluation Summary}
Our experimental evaluation demonstrates that:
(1) \Tool{} can compute higher moments for a rich class of probabilistic loops with various characteristics,
(2) \Tool{} outperforms the state-of-the-art of moment computation for probabilistic loops in terms of supported programs and efficiency, and
(3) \Tool{} computes exact moments magnitudes faster than sampling can establish reasonable approximations.

%% file: 08-related.tex
\section{Related Work}\label{sec:related}

Using recurrence equations to extract closed-forms for variables and quantitative invariants of loops is a well-studied technique for non-probabilistic programs~\cite{Farzan2015, Breck2020, Kincaid2019, Kincaid2018, Oliveira2016, Humenberger2017, Humenberger2018, Kovacs2008, Rodriguez-carbonell2004}.
Because a classical program is a special case of a probabilistic program, our technique presented in Section~\ref{sec:computing-moments} is a generalization of the closed-form computation for classical programs to probabilistic programs.
Moreover, the generalization to probabilistic programs is not trivial. One reason for this is that for classical programs the closed form for $x^p$ is just the closed-form for x to the power $p$. However, this fails for moments, as in general $E(x^p)$ is not equal to $E(x)^p$.

A common approach to quantitatively and exactly analyze probabilistic programs is to employ probabilistic model checking techniques~\cite{Baier2008, Kwiatkowska2011, Dehnert2017, Katoen2011, Holtzen2021}.

Exact inference for computing precise posterior distributions for probabilistic programs has been studied in \cite{Gehr2016, Holtzen2020, Narayanan2016, Claret2013, Saad2021}.
An interesting direction for future research is using our techniques to assist probabilistic inference in the presence of loops.

A different  approach to characterize the distributions of program variables are statistical methods such as Monte Carlo and hypothesis 
testing~\cite{Younes2006}.
Simulations are however performed on a chosen finite number of program steps and do not provide guarantees over a potentially infinite execution, such as unbounded loops, limiting thus their use (if at all) for invariant generation.

In~\cite{McIver2005}, a deductive approach, the \emph{weakest pre-expectation calculus}, for reasoning about PPs with discrete program variables is introduced.
Based on the weakest pre-expectation calculus, \cite{Katoen2010} presents the first template-based approach for generating linear quantitative invariants for PPs.
Other works~\cite{Feng2017,Chen2015} also address the synthesis of non-linear invariants or employ \emph{martingale} expressions~\cite{Barthe2016}.
All of these works target a slightly different problem and, unlike our approach, rely on templates.
The first data-driven technique for invariant generation for PPs is presented in \cite{Bao2021}.

Another line of related work comes with computing bounds over expected values~\cite{Bouissou2016,Karp1994,Chatterjee2020} and higher moments~\cite{Kura2019,Wang2021}.
The approach in~\cite{Bouissou2016} can provide bounds for higher moments and can handle non-linear terms at the price of producing more conservative bounds.
In contrast, our approach natively supports probabilistic polynomial assignments and provides a precise symbolic expression for higher moments.

The technique presented in~\cite{Bartocci2019} automates the generation of so-called moment-based invariants for a subclass of PPs with polynomial probabilistic updates and sets the basis for fully automatic exact higher moment computation.
Relative to our approach, \cite{Bartocci2019} supports neither if-statements (thus also no guarded loops), state-dependent distribution parameters, nor circular variable dependencies.
Our work establishes stronger theoretical foundations.

%% file: 09-conclusion.tex
\section{Conclusion}\label{sec:conclusion}

We describe a fully automated approach for inferring exact higher moments for program variables of a large class of probabilistic loops with complex control flow, polynomial assignments, symbolic constants, circular dependencies among variables, and potentially uncountable state spaces. Our work uses program transformations to normalize and simplify probabilistic programs while preserving the joint distribution of program variables. We propose a power reduction technique for finite program variables to ease the complex polynomial arithmetic of probabilistic programs. We prove soundness and completeness of our approach, by establishing the theory of moment-computable probabilistic loops. We demonstrate use cases of exact higher moments in the context of computing tail probabilities and recovering distributions from moments.  Our experimental evaluation illustrates the applicability of our work, solving several examples whose automation so far was not yet supported by the state-of-the-art in probabilistic program analysis.